\documentclass[a4paper, 11pt, leqno]{article}
\usepackage[english]{babel}
\usepackage[T1]{fontenc}
\usepackage{a4wide}
\usepackage{amsmath,amssymb,amsthm}
\usepackage{graphicx,latexsym,color,hyperref}
\usepackage{newpxtext, newpxmath}
\usepackage[title]{appendix}
\usepackage{rotating, comment}
\usepackage{setspace}
\usepackage[authoryear]{natbib}

\newif\ifdraft

\drafttrue

\definecolor{jpred}{rgb}{0.89,0.3,0.4}

\theoremstyle{plain}
\newtheorem{theorem}{Theorem}[section]

\newtheorem{lemma}[theorem]{Lemma}
\newtheorem{corollary}[theorem]{Corollary}

\theoremstyle{remark}
\newtheorem{remark}[theorem]{Remark}
\theoremstyle{definition}
\newtheorem{definition}[theorem]{Definition}

\theoremstyle{remark}
\newtheorem{XxmpX}{Example} 
\newenvironment{ex}    
  {%
   \pushQED{\qed}\begin{XxmpX}}
  {\popQED\end{XxmpX}}

  \makeatletter
\newcommand\assumptionlabel[1]{\hspace\labelsep
                               \normalfont\bfseries #1\ \ \gdef\@currentlabel{#1}}
\newenvironment{assumption}
               {\smallskip\list{}{\labelwidth\z@ \itemindent-\leftmargin
                        }}
               {\endlist}
\makeatother

\newcounter{AlgorithmJP}[section]

\renewcommand{\theAlgorithmJP}{\thesection.\arabic{AlgorithmJP}}

\def\cD{\mathcal{D}}
\def\cM{\mathcal{M}}
\def\cN{\mathcal{N}}

\def\0{\mathbf{0}}
\def\er{\mathbb{R}}
\def\prob{\mathbb{P}}
\def\ee{E}
\def\X{\mathbb{X}}
\def\cQ{\mathcal{Q}}
\def\L{\mathcal{L}}
\def\cvar{\mathop{\rm CVaR}\nolimits}

\def\argmax{\mathop{\rm arg\, max}}
\def\reint{\mathop{\rm reInt}}
\def\conv{\mathop{\rm conv}}
\def\lin{\mathop{\rm lin}}
\def\rank{\mathop{\rm rank}}
\def\reint{\mathop{\rm ri}}
\renewcommand\triangledown{\nabla}
\newcommand\alloc[1]{\mathbb{#1}}

\def\eop{\hfill{$\Box$}\medskip}

\newcommand\ind[1]{{1\kern-0.4em 1}_{\{#1\}}}

\newcommand\var{\textrm{VaR}}

\newcommand\bX{X}
\newcommand\bx{x}
\newcommand\br{R}
\newcommand\bR{R}

\title{On the solution uniqueness in portfolio optimization and risk analysis\footnote{The first author (BG) thanks the University of Leicester for granting him the academic study leave to do this research. The research of the second (AP) and third (JP) author was supported by the National Science Centre, Poland, under Grant 2014/13/B/HS4/00176.
}}
\author{
Bogdan Grechuk\footnote{Department of Mathematics, University of Leicester, UK (e-mail: bg83@leicester.ac.uk)}
\and
Andrzej Palczewski\footnote{Faculty of Mathematics, University of Warsaw, Banacha 2, 02-097 Warszawa, Poland (e-mail: A.Palczewski@mimuw.edu.pl)}
\and
Jan Palczewski\footnote{School of Mathematics, University of Leeds, Leeds LS2 9JT, UK (e-mail: J.Palczewski@leeds.ac.uk)}}

\date{\today}

\begin{document}

\maketitle

\begin{abstract}

We consider the issue of solution uniqueness for portfolio optimization problem and its inverse for asset returns with a finite number of possible scenarios. The risk is assessed by deviation measures introduced by
[Rockafellar et al.,  Mathematical Programming, Ser. B, 108 (2006), pp.~515--540]
instead of variance as in the Markowitz optimization problem. We prove that in general one can expect uniqueness neither in forward nor in inverse problems. We discuss consequences of that non-uniqueness for several problems in risk analysis and portfolio optimization, including capital allocation, risk sharing, cooperative investment, and the Black-Litterman methodology. In all cases, the issue with non-uniqueness is closely related to the fact that subgradient of a convex function is non-unique at the points of non-differentiability.  
We suggest methodology to resolve this issue by identifying a unique ``special'' subgradient satisfying some natural axioms. This ``special'' subgradient happens to be the Steiner point of the subdifferential set.   
\end{abstract}


\section{Introduction}

In various problems in economics and finance, including capital allocation \citep{kalkbrener2005}, risk sharing \citep{filipovic2008}, cooperative investment \citep{grechuk2015synergy}, inverse portfolio problem \citep{BGP2012}, and generalized Black-Litterman model \citep{palczewski2016}, it is important to identify a \emph{unique} solution. 
We show that a solution to many of such problems can be expressed in an explicit way using a sub-gradient of some convex function $f:{\mathbb R}^N \to {\mathbb R}$ at some point $X \in {\mathbb R}^N$, and the solution is unique if and only if $f$ is differentiable at $X$. Because every convex function $f$ is differentiable almost everywhere \cite[Theorem 25.5]{rockafellar1970}, one may expect that such a solution should be unique in all ``practical'' cases. While this is indeed true in the context of risk sharing, we demonstrate that this intuition fails badly in other contexts. To resolve this issue, we suggest an axiomatic framework for selecting a unique special sub-gradient, which we call an \emph{extended gradient}, from the subdifferential set $\partial f(X)$ of every convex function $f$ at any point $X$. In fact, our extended gradient coincides with the Steiner point \citep{schneider1971} of $\partial f(X)$. This allows us to resolve the issue of solution non-uniqueness in various applications. 

Capital allocation problem is one of the basic problems in risk management, which has been studied in a number of papers, see e.g. \cite{denault2001coherent}, \cite{kalkbrener2005}, \cite{dhaene}, and references therein. The problem is to distribute the risk capital among $n$ subsidiaries or business units. Equivalently (see \cite{cherny2011two}), the problem is to decide how much the risk coming from each subsidiary contributes to the total (cumulative) risk. \cite{kalkbrener2005} established necessary and sufficient conditions on the risk measure for the existence of capital allocation with two highly desirable properties: linearity and diversification. Unfortunately, linear diversifying capital allocation may be non-unique, and, in this case, it is unclear which one to select. \cite{cherny2011two} suggested an additional ``law-invariance'' axiom, under which the capital allocation becomes unique for some specific family of risk measures, but not in general. \cite{Grechuk2015} introduced so-called ``centroid capital allocation'', which is unique but lacks axiomatic foundation. The present work suggests a capital allocation approach based on the Steiner point of a sub-differential set, which is always unique and follows from some natural axioms.

Optimal risk sharing, originated by \cite{borch1962}, \cite{arrow1963}, and others, is a classical problem which asks for the optimal redistribution of risk among $n$ agents. Such redistribution  is called Pareto optimal if no agent can decrease their risk without increasing the risk for some other agents. If agents are allowed to trade, they will eventually arrive at some special Pareto optimal allocation, which is called equilibrium allocation \citep{filipovic2008}. However, if equilibrium allocation is not unique, which one to choose? Our Steiner point approach can be applied to this problem as well.    

In the problem of cooperative investment, $m$ agents decide that instead of investing individually, they can form a coalition, buy a joint portfolio, and then distribute the profit of this joint portfolio in the same way as in the optimal risk sharing problem, see e.g. \cite{xia2004multi} and \cite{grechuk2013}.  
The utility of investor $i$ is $U_i(Z_i)$, where $U_i$ is some utility function and $Z_i$ the random wealth of agent $i$ at the investment horizon. \cite{grechuk2015synergy} show that, under some mild conditions on $U_i$, cooperative investment is strictly preferable for all agents compared to their optimal individual investment strategies. In the cooperative investment problem, the coalition's preferences can be represented by a cooperative utility function $U^*$. The coalition solves an optimization problem with the utility $U^*$ to find an optimal portfolio with the terminal wealth $X^*$. This terminal wealth must consequently be distributed among investors: one has to find a Pareto optimal allocation $(Z_1, \dots,  Z_m )$ such that $X^* = \sum_{i=1}^m Z_i$. There are usually infinitely many Pareto-optimal wealth allocations, but \cite{grechuk2015synergy} defined an allocation which can be considered as ``fair''. The issue is that this ``fair'' allocation is, in general, non-unique, as we demonstrate in this paper. This issue is of fundamental importance, because different allocation methods may favour different agents.
Because this non-uniqueness is the consequence of possible non-differentiability of $U^*$, this issue is resolved by our Steiner point approach provided that $U^*$ is a concave function.
 
In the realm of portfolio analysis we consider a market with a riskless asset and $n$ risky assets. Portfolios are represented as combinations $x_1R^{(1)}+\dots +x_n R^{(n)}$, where the vector random variable $R = (R^{(1)}, \ldots, R^{(n)})^T$  denotes excess returns of risky assets. The objective is to find a portfolio allocation (fractions of wealth invested in the risky assets) $x = (x_1, \ldots, x_n)^T$ that solves the following optimization problem:
\begin{equation}\label{eqn:intro}
 \min_x \rho(R^Tx) \quad \textrm{subject to } \mu^T x\ge \Delta,
\end{equation}
where $\rho$  measures portfolio risk,  $\Delta$ is the target excess return and 
\[
\mu =(\mu_1, \ldots, \mu_n)^T = (\ee[R^{(1)}], \ldots, \ee[R^{(n)}])^T. 
\]

In this paper we study the uniqueness of solutions to problem \eqref{eqn:intro} and solutions to the following inverse problem: given a vector $x^*$,  the information on the distribution of $R$ sufficient to compute $\rho(R^Tx)$ for any $x$, and $\Delta >0$ find a vector of mean returns $\mu$ such that $x^*$ is a solution to problem \eqref{eqn:intro} for that $\mu$. Notice that the inverse problem we are interested in is meaningful only when the risk measure $\rho$ is indifferent to the location parameter of distribution $R$, e.g., standard deviation, variance of portfolio returns, or a deviation measure of \cite{rockafellar2006b}.  

The problem of portfolio inverse optimization  under different formulations has been investigated by several authors. \cite{BGP2012} considers an inverse optimization in a robust optimization framework with the portfolio mean as the objective function and risk accounted for in constraints.  The problem of uniqueness is not addressed in that paper, particularly because under their assumption of normality of asset returns the forward problem always has a unique solution. Due to the number of degrees of freedom (in the mean-variance case it is both the mean, variance and the target return $\Delta$ that are to be inferred from the optimal portfolio), the inverse problem inherently has many solutions. \citet{GZ2014, GZ2016} attempt to infer risk preferences of an investor: assuming a complete knowledge of the distribution of $R$ and portfolio $x^*$, they look for a risk measure $\rho$ for which $x^*$ is an optimal solution to \eqref{eqn:intro}. They solved this inverse problem for two classes of risk measures $\rho$: deviation measures and coherent risk measures.  

A motivation for analyzing the uniqueness of forward and inverse optimization problems stated above comes from the Black-Litterman asset allocation model, cf. \cite{BlackLitterman91} where the model is formulated and \cite{litterman} for a more detailed presentation. In the classical Black-Litterman model, the risk is modeled by the variance. The inverse optimization, used to establish the equilibrium distribution, has a unique solution. The variance, however,  is a poor measure of risk for non-Gaussian distributions. \cite{rockafellar2006b} promotes deviations measures which are rooted in coherent risk measures but are indifferent to the location parameter of the distribution (as the variance). The optimization problem \eqref{eqn:intro} retains its convexity in $x$, but the uniqueness of solutions to the forward and inverse problems has not been studied. A general theory of convex optimization implies that they depend on the interplay between the distribution of $R$ and the risk measure $\rho$. Our Steiner point approach can be used to identify a unique ``special'' solution to this problem as well.    

In the context of asset management, many papers assume a finite (but possibly large) number of scenarios for the future excess return $R$ (for example a historical time series of asset returns) and this is the case that we research in this paper. The reader is referred to, e.g., \cite{KPU2002, fabozzi2010robust, LIM2011163, grechuk2017direct} for theoretical and finance-centred contributions and \cite{gaivoronski2004, lim2010portfolio, lwin2017mean} for numerical methods; applications outside of finance can be found in the monograph \cite{conejo2010decision} and references therein. Although the question of  existence of optimal solutions has been solved, the problem of uniqueness for a finite number of scenarios has not been analyzed carefully enough. We perform detailed analysis of that problem for arbitrary discrete scenarios and a class of deviation measures which we call ``finitely generated risk measures'' which includes Conditional Value-at-Risk (CVaR), mixed CVaR and mean absolute deviation. In our approach we use the characterization of deviation measures by their risk envelopes introduced in \cite{rockafellar2006c}.

Our contributions are based on a new link between the uniqueness of an optimal portfolio $x^*$ in \eqref{eqn:intro} and the number of risk identifiers for the deviation measure $\rho(R^T x^*)$. This has three consequences. Firstly, the portfolio optimization problem has a unique solution for any $\mu \in \er^n$ that does not belong to a finite number of hyperplanes; therefore, for practical applications the uniqueness can be safely assumed. Secondly, a unique optimal portfolio corresponds to many risk identifiers and, consequently, there are many Pareto-optimal sharing arrangements in cooperative investment, which is obviously highly inconvenient in practice. It is also surprising as this possibility was only inferred from the general convexity theory and treated as an unlikely and inconvenient case that is not of prime importance, see \cite{grechuk2015synergy}. The third consequence is related to the extension of the Black-Litterman model to arbitrary distributions and deviation measures \citep{palczewski2016}. Analogously as in the classical model   the first step of the extended model is to solve an inverse portfolio problem in which for a market (or benchmark) portfolio $x^*$ one establishes an equilibrium mean return $\mu_{eq}$ that yields $x^*$ as an optimal solution, cf. \citet[Section 4]{palczewski2016}. We demonstrate that if $x^*$ is a unique optimal solution for a particular $\mu^*$ then the inverse problem has multiple solutions. Hence, the final investment recommendation coming out of Black-Litterman methodology is not unique. 
Our Steiner point approach is then used to select a unique recommendation.

The rest of the paper is organized as follows. Section \ref{sec:gradient} suggests a Steiner point approach for assigning a unique ``extended gradient'' of every convex function on ${\mathbb R}^N$ at every point. Section \ref{sec:capital} applies this methodology for selecting the unique solution in the capital allocation and risk sharing problems. Section \ref{sec:portf} formulates the portfolio optimization problem in the framework of deviation measures, defines portfolio risk generators and discusses the portfolio uniqueness problem in terms of portfolio risk generators. 
Section \ref{sec:7} formulates the cooperative investment problem and resolves the issue of non-uniqueness of its solution.
Section \ref{sec:5} discusses the dichotomy between uniqueness of solutions of the forward and inverse optimization problems. 
Section \ref{sec:6} considers consequences of non-uniqueness for the Black-Litterman model for non-Gaussian distributions. Section \ref{sec:concl} concludes the work.

\section{Extended gradient of a convex function}\label{sec:gradient}

\subsection{Definition and axiomatic characterization}

Let $f: {\mathbb R}^n \to {\mathbb R}$ be an arbitrary (finite valued) convex function. It is known \citet[Theorem 23.1]{rockafellar1970} that the one-sided limit
\begin{equation}\label{eqn:1sided}
\phi_{f,Y}(X) = \lim\limits_{\epsilon\to 0^+} \frac{f(Y+\epsilon X) - f(Y)}{\epsilon}
\end{equation}  
exists for every $X,Y \in {\mathbb R}^n$. Limit $\phi_{f,Y}(X)$ is called the \emph{directional derivative of $f$ at $Y$ with respect to $X$}. 
We say that $f$ is (G\^{a}teaux) differentiable at $Y \in {\mathbb R}^n$ if the (two-sided) limit $\lim_{\epsilon\to 0} \frac{f(Y+\epsilon X) - f(Y)}{\epsilon}$ exists for every $X \in {\mathbb R}^n$. In this case, $\phi_{f,Y}(X)$ is a 
linear functional in $X$, and can be represented as $\phi_{f,Y}(X) = Q^TX$ for some $Q \in {\mathbb R}^n$, which is usually denoted as $Q=\nabla f(Y)$ and called the \emph{gradient} of $f$ at $Y$.
It is known \citet[Theorem 25.5]{rockafellar1970} that any finite-valued convex function $f$ on ${\mathbb R}^n$ is differentiable almost everywhere. 

This section develops an axiomatic framework for ``extending'' the notion of a gradient in such a way that the ``extended gradient''
is defined for \emph{every} convex function $f: {\mathbb R}^n \to {\mathbb R}$ at \emph{every} point $Y \in {\mathbb R}^n$. 
In the next section we will demonstrate that this ``extended gradient'' is useful in financial applications, including capital allocation, risk sharing, and cooperative investment.

Let ${\cal F}$ be a set of all convex functions $f: {\mathbb R}^n \to {\mathbb R}$. Formally, we define \emph{extended gradient} as a map $G:{\cal F} \times {\mathbb R}^n \to {\mathbb R}^n$, which assigns to every $f \in {\cal F}$ and $Y \in {\mathbb R}^n$ a vector $G_Y(f) \in {\mathbb R}^n$, such that the following properties hold: 
\begin{assumption}
\item[(G1)] Additivity: $G_Y(f + g) = G_Y(f) + G_Y(g)$ for all $f,g \in {\cal F}$ and all $Y \in {\mathbb R}^n$;
\item[(G2)] Rotation invariance: Let $f \in {\cal F}$ and $g(Y)=f(AY), \, Y \in {\mathbb R}^n$, where $A$ is an $n \times n$ \emph{rotation matrix}, that is, matrix such that $A^T = A^{-1}$ and $\det(A)=1$. Then 
$$
G_Y(g) = A^{-1} G_{AY}(f), \quad \forall Y \in {\mathbb R}^n.
$$
\item[(G3)] Continuity: Let $Y \in {\mathbb R}^n$, $f \in {\cal F}$, and $f_1, f_2, \ldots$ be a sequence of functions in ${\cal F}$ such that $\lim\limits_{m \to \infty} \phi_{f_m,Y}(X) = \phi_{f,Y}(X)$ for all $X \in {\mathbb R}^n$. Then 
$$
\lim\limits_{m \to \infty} G_Y(f_m) = G_Y(f). 
$$
\item[(G4)] Linear differentiation: Let $Q \in {\mathbb R}^n$, and let $f(Y)=Q^T Y, \, \forall Y \in {\mathbb R}^n$ be a linear function. Then 
$$
G_Y(f) = Q, \quad \forall Y \in {\mathbb R}^n.
$$
\end{assumption}

Properties (G1)-(G4) are desirable properties for any extension of the concepts of ``derivative'' or ``gradient''. (G1) states that the derivative/gradient of a sum is the sum of derivatives/gradients of summands, (G2) is an invariance under rotation of the coordinate system, (G3) is a manifestation of the fact that derivative/gradient is a \emph{local} property of a function at a point, and two functions which ``look locally almost the same'' in every direction should have ``almost identical'' gradients. Finally, (G4) states that the derivative/gradient of a linear function is a constant. Theorem \ref{thm:extgrad} below states that, somewhat surprisingly, these natural properties are sufficient for the \emph{unique characterization} of $G$.

Directional derivative $\phi_{f,Y}(X)$ can be represented in the form 
\begin{equation}\label{eqn:support}
\phi_{f,Y}(X) = \sup_{Q \in \partial f(Y)} Q^T X
\end{equation}
see \citet[Theorem 23.4]{rockafellar1970}, where $\partial f(Y)$ is called \emph{subdifferential} of $f$ at $Y$, and is defined as a set of all $Q \in {\mathbb R}^n$ such that $f(X) \geq f(Y) + Q^T (X-Y), \, \forall X \in {\mathbb R}^n$. Set $\partial f(Y)$ is always non-empty, convex, and compact, see \citet[Theorem 23.4]{rockafellar1970}. Let ${\cal K}$ be the family of all non-empty convex compact subsets of ${\mathbb R}^n$.

With $f_1=f_2=\dots=f_m=\dots=g$, property (G3) implies that $G_Y(f) = G_Y(g)$ whenever $\phi_{f,Y}(X) = \phi_{g,Y}(X)$ for all $X \in {\mathbb R}^n$. Equivalently, $G_Y(f) = G_Y(g)$ whenever $\partial f(Y) = \partial g(Y)$. Hence $G_Y(f)$ can be represented as
\begin{equation}\label{eqn:subdifrepr}
G_Y(f) = S(\partial f(Y)),
\end{equation}
where $S$ is a map assigning to every set $K \in {\cal K}$ a vector $S(K) \in {\mathbb R}^n$.

Properties (G1)-(G4) of $G_Y(f)$ can be equivalently written as properties of the map $S$. 
For any $K_1, K_2 \subset {\mathbb R}^n$, the set $K_1 + K_2 = \{Q_1 + Q_2\,|\,Q_1 \in K_1, \, Q_2 \in K_2\} $ is called (Minkowski) sum of $K_1$ and $K_2$. Theorem 23.8 in \cite{rockafellar1970} implies that $\partial((f + g)(Y)) = \partial f(Y) + \partial g(Y)$ for all $f,g \in {\cal F}$ and all $Y \in {\mathbb R}^n$. Hence, property (G1) is equivalent to
\begin{assumption}
\item[(S1)] $S(K_1 + K_2) = S(K_1) + S(K_2)$ for all $K_1 \in {\cal K}, K_2 \in {\cal K}$.
\end{assumption}

Property (G4) is equivalent to $S(\{Q\})=Q$. Substituting this into (S1), we get $S(K + Q) = S(K) + Q$. In other words, if the set $K$ is translated by a vector $Q \in {\mathbb R}^n$, $S(K)$ is translated by the same vector.

Let $A$, $f$, and $g$ be as defined in (G2). Theorem 23.9 in \cite{rockafellar1970} implies that $\partial(g(Y)) = A^{-1} \partial f(AY)$. Hence, property (G2) is equivalent to 
$
S(AK) = AS(K), \, \forall K \in {\cal K}.
$
This implies that $S(AK + Q) = AS(K) + Q$ for all $Q \in {\mathbb R}^n$, or, equivalently,
\begin{assumption}
\item[(S2)] $S(TK) = T S(K)$ for all $K \in {\cal K}$ and all transformations $T:{\mathbb R}^n \to {\mathbb R}^n$ in the form $T(X) = AX+Q$, where $A$ is a rotation matrix, and $Q \in {\mathbb R}^n$. Such transformations $T$ are called \emph{rigid body motions}.
\end{assumption}

For every non-empty closed convex set $K$ in ${\mathbb R}^n$, its \emph{support function} is given by $f_K(X)=\sup\{Q^TX, | \, Q \in K\}$. In particular, \eqref{eqn:support} implies that directional derivative $\phi_{f,Y}(X)$ is a support function of the subdifferential $\partial f(Y)$. For sets $K, K_1, K_2, \ldots$ in ${\cal K}$, a combination of Corollary P4.A and Corollary 3A in \cite{salinetti1979}
implies that point-wise convergence of the support function of $K_m$ to the support functions of $K$ is equivalent to $\lim_{m \to \infty}h(K_m,K)=0$, where $h$ denotes the Hausdorff 
distance\footnote{The Hausdorff distance $h(K,L)$ between any subsets $K$ and $L$ of ${\mathbb R}^n$ is defined as $h(K,L) = \max\{\sup_{X \in K} \inf_{Y \in L}d(X,Y), \sup_{Y \in L} \inf_{X \in K}d(X,Y)\}$, where $d(.,.)$ in the usual Euclidean distance in ${\mathbb R}^n$.} 
between sets. This implies the following reformulation of property (G3).

\begin{lemma}\label{lem:g3-s3}
Let $S:{\cal K}\to{\mathbb R}^n$, and let $G_Y$ be given by \eqref{eqn:subdifrepr}. Then $G_Y$ satisfies (G3) if and only if $S$ satisfies
\begin{assumption}
\item[(S3)] Map $S$ is continuous with respect to the Hausdorff metric. That is, 
$$
\lim\limits_{m \to \infty} S(K_m) = S(K) 
$$
whenever sets $K, K_1, K_2, \ldots \in {\cal K}$ are such that $\lim\limits_{m \to \infty}h(K_m,K)=0$. 
\end{assumption}
\end{lemma}
\begin{proof}
First, assume that (S3) holds, $Y \in {\mathbb R}^n$, $f \in {\cal F}$, and $f_m$ is a sequence of functions as in (G3). Let $K = \partial f(Y)$, $K_m = \partial f_m(Y)$. Then $\phi_{f,Y}(X)$ and $\phi_{f_m,Y}(X)$ are the support functions of $K$ and $K_m$, respectively, and condition $\lim_{m \to \infty} \phi_{f_m,Y}(X) = \phi_{f,Y}(X)$ implies that $\lim\limits_{m \to \infty}h(K_m,K)=0$. Then, by (S3), $\lim\limits_{m \to \infty} S(K_m) = S(K)$, which, by \eqref{eqn:subdifrepr}, is translated to $\lim_{m \to \infty} G_Y(f_m) = G_Y(f)$ and proves (G3).

Conversely, assume that (G3) holds and let $K, K_1, K_2, \ldots \in {\cal K}$ be such that $\lim_{m \to \infty}h(K_m,K)=0$. Let $f_0, f_1,f_2, \ldots$ be the support functions of these sets. Then $\lim_{m \to \infty} f_m = f_0$ point-wise. Because each $f_m$ is positively homogeneous, its directional derivative at $Y=0$ is 
\[
\phi_m(X) = \lim_{\epsilon\to 0^+} \frac{f_m(0+\epsilon X) - f_m(0)}{\epsilon} = f_m(X). 
\]
Hence, $\lim_{m \to \infty} \phi_m = \phi_0$ point-wise, and (G3) with $Y=0$ implies that $\lim_{m \to \infty} G_0(f_m) = G_0(f)$. With \eqref{eqn:subdifrepr}, this translates to $\lim_{m \to \infty} S(K_m) = S(K)$ and proves (S3).     
\end{proof}

In summary, $G_Y(f)$ satisfies properties (G1)-(G4) if and only if it is representable in the form \eqref{eqn:subdifrepr}, with a map $S:{\cal K}\to{\mathbb R}^n$ satisfying (S1)-(S3). However,  
Theorem 1 in \cite{schneider1971} states that, in any dimension $n \geq 2$, there is a \emph{unique} map $S$ satisfying   (S1)-(S3), and it is given by
\begin{equation}\label{eqn:steiner1}
S(K) = \frac{n}{|S^{n-1}|}\int_{S^{n-1}} X f_K(X) dX,
\end{equation}
where $S^{n-1}=\{X \in {\mathbb R}^n\,|\,||X||=1\}$ denotes the unit sphere in ${\mathbb R}^n$, $|S^{n-1}|$ is its surface area, and $f_K(X)$ is the support function of $K$. $S(K)$ is known as \emph{Steiner point} of the set $K$. Equivalently (see e.g. \cite{dentcheva1998}),
\begin{equation}\label{eqn:steiner2}
S(K) = \frac{1}{|B_1|}\int_{B_1} \nabla f_K(X) dX,
\end{equation}
where $B_1=\{X \in {\mathbb R}^n\,|\,||X|| \leq 1\}$ denotes the unit ball, $\nabla$ is the gradient, and the integral is well-defined because the support function of any $K \in {\cal K}$ is differentiable almost everywhere. If $K = \partial f(Y)$, its support function is $\phi_{f,Y}(X)$, and we obtain
\begin{equation}\label{eqn:extgrad}
G_Y(f) = S(\partial f(Y)) = \frac{1}{|B_1|}\int_{B_1} \nabla \phi_{f,Y}(X) dX. 
\end{equation} 
We will summarise the above discussion in the following theorem. 
\begin{theorem}\label{thm:extgrad}
In any dimension $n \geq 2$, the extended gradient $G_Y(f)$ is \emph{uniquely} characterized by properties (G1)-(G4), and it is given by \eqref{eqn:extgrad}, where $\phi_{f,Y}(X)$ is defined in \eqref{eqn:1sided}.
\end{theorem}

The next theorem provides an alternative formula for the extended gradient $G_Y(f)$.
\begin{theorem}\label{thm:extgrad2}
For every convex function $f: {\mathbb R}^n \to {\mathbb R}$, and every $Y \in {\mathbb R}^n$, the extended gradient $G_Y(f)$ is given by 
\begin{equation}\label{eqn:extgrad2}
G_Y(f) = \lim\limits_{\epsilon \to 0^+}\frac{1}{|B_\epsilon(Y)|}\int_{B_\epsilon(Y)} \nabla f(X)\, 1_{\{X \in D\}} dX,
\end{equation} 
where $B_\epsilon(Y)=\{X \in {\mathbb R}^n\,|\,||X-Y|| \leq \epsilon\}$ is the ball centred at $Y$ with radius $\epsilon$, $\nabla$ is the (almost everywhere defined) gradient of $f$, $D = \{X: \partial f(X) = \{ \nabla f(x) \} \}$ and the limit is guaranteed to exist. 
\end{theorem}
\begin{proof}\newcommand\PH[1]{\phi_{f,Y}(#1)}
Let $\mathcal{S} \subset S^{n-1}$ be the set of such $Z$ that
\begin{itemize}
 \item[(a)] $Y + \epsilon Z \in D$ for almost all $\epsilon$,
 \item[(b)] $\nabla\PH{Z}$ is a singleton.
\end{itemize}
We shall show that $\mathcal{S}$ is of $\sigma^n$-full measure, where $\sigma^n$ is the spherical measure on $S^{n-1}$. Statement (a) follows from the fact that $D$ has a full Lebesgue measure. For (b), the function $g(Z):=\PH{Z}$ is convex, positively homogeneous (as the support function of $\partial f(Y)$) and $\nabla \PH{Z} = \partial g(Z)$. By \citet[Theorem 25.5]{rockafellar1970}, $\partial g(Z)$ is a singleton for almost every $Z \in {\mathbb R}^n$. Because $g$ is positively homogeneous, the same is true for $\sigma^n$-almost every $Z$ on $S^{n-1}$, which completes the proof that $\mathcal{S}$ is of full measure.

We have 
\begin{multline}\label{eqn:pr1}
\frac{1}{|B_\epsilon(Y)|}\int_{B_\epsilon(Y)}(\triangledown f(X)-\triangledown\phi_{f,Y}(X-Y)) \, 1_{\{\frac{X-Y}{||X-Y||} \in \mathcal{S}\}}\,dX\\
=
\int_{S^{n-1}} 1_{\{Z \in \mathcal{S}\}} \frac{n}{\epsilon^n} \int_0^\epsilon (\triangledown f(Y + \alpha Z)-\triangledown\phi_{f,Y}(\alpha Z)) \alpha^{n-1} 1_{\{Y + \alpha Z \in D\}}\, d\alpha\,\sigma^n(dZ).
\end{multline}
\citet[Corollary 23.5.3]{rockafellar1970} states that
\[
\triangledown\phi_{f,Y}(Z)) = \argmax_{Q \in \partial f(Y)} Q^T Z,
\]
i.e., these are points of $\partial f(Y)$ at which $Z$ is normal. The above formula implies also that 
\begin{equation}\label{eqn:qr1}
\triangledown\phi_{f,Y}(Z) = \triangledown\phi_{f,Y}(\lambda Z) \qquad \text{for $\lambda > 0$.}
\end{equation}

By \citet[Theorem 24.6]{rockafellar1970}, for $Z \in \mathcal{S}$ we have
\[
\lim_{\substack{\epsilon \to 0^+\\ Y + \epsilon Z \in D}} \nabla f(Y + \epsilon Z) = \nabla \PH{Z}.
\]
Since the set $\bigcup_{X \in B_1(Y)} \partial f(X)$ is bounded by \citet[Theorem 24.7]{rockafellar1970}, the dominated convergence theorem implies the limit
\begin{multline*}
\lim_{\epsilon \to 0^+} \frac{n}{\epsilon^n} \int_0^\epsilon (\triangledown f(Y + \alpha Z)-\triangledown\phi_{f,Y}(Z)) \alpha^{n-1} 1_{\{Y + \alpha Z \in D\}}\, d\alpha\\
=
\lim_{\epsilon \to 0^+} \frac{n}{\epsilon} \int_0^\epsilon (\triangledown f(Y + \alpha Z)-\triangledown\phi_{f,Y}(Z)) \frac{\alpha^{n-1}}{\epsilon^{n-1}} 1_{\{Y + \alpha Z \in D\}}\, d\alpha = 0.
\end{multline*}
Using \eqref{eqn:qr1} and the dominated convergence theorem this further implies that
\[
\lim_{\epsilon \to 0^+} \int_{S^{n-1}} 1_{\{Z \in \mathcal{S}\}} \frac{n}{\epsilon^n} \int_0^\epsilon (\triangledown f(Y + \alpha Z)-\triangledown\phi_{f,Y}(\alpha Z)) \alpha^{n-1} 1_{\{Y + \alpha Z \in D\}}\, d\alpha\,\sigma^n(dZ) = 0.
\]
Recalling \eqref{eqn:pr1} and that $\mathcal{S}$ is of $\sigma^n$-full measure, this leads to
\[
G_Y(f) = \frac{1}{|B_\epsilon(Y)|}\int_{B_\epsilon(Y)} \triangledown\phi_{f,Y}(X-Y) \, 1_{\{\frac{X-Y}{||X-Y||} \in \mathcal{S}\}}\, dX.
\]
Invoking again \eqref{eqn:qr1}, the right hand side of the expression above equals
\[
\frac{1}{|B_1(Y)|}\int_{B_1(Y)} \triangledown\phi_{f,Y}(X-Y)\, 1_{\{\frac{X-Y}{||X-Y||} \in \mathcal{S}\}}\, dX = \frac{1}{|B_1|}\int_{B_1} \triangledown\phi_{f,Y}(X) 1_{\{\frac{X}{||X||} \in \mathcal{S}\}} dX.
\]
which completes the proof due to \eqref{eqn:extgrad} and the fact that $\int_{B_1} 1_{\{\frac{X}{||X||} \in \mathcal{S}\}} dX = 0$.

\end{proof}

Theorem \ref{thm:extgrad2} gives a nice intuitive interpretation of $G_Y(f)$: it is an average gradient of $f$ in a small ball centred in $Y$ when the radius of the ball goes to $0$. 
While Theorem \ref{thm:extgrad} is applicable in dimension $n\geq 2$, the extended gradient defined by formulas \eqref{eqn:extgrad}-\eqref{eqn:extgrad2} is well-defined in dimension $n=1$ as well.
For a convex function $f: {\mathbb R} \to {\mathbb R}$, and every $y \in {\mathbb R}$, $G_y(f)$ in \eqref{eqn:extgrad}-\eqref{eqn:extgrad2} is given by 
\begin{equation}\nonumber
G_y(f) = \frac{f'_+(y)+f'_-(y)}{2},
\end{equation} 
where $f'_+(y)$ and $f'_-(y)$ are right and left derivatives of $f$ at $y$, respectively.

Because Steiner point of any set $K \in {\cal K}$ belongs to $K$, \eqref{eqn:extgrad} implies that $G_Y(f) \in \partial f(Y)$, or, in words, extended gradient always belongs to the subdifferential set. In particular, $G_Y(f) = \triangledown f(Y)$, whenever the latter exists. This justifies the name \emph{extended gradient} for $G_Y(f)$.

\subsection{Discussion}

Let us discuss the extended gradient \eqref{eqn:extgrad}-\eqref{eqn:extgrad2} from a different perspective. Because every convex function $f:{\mathbb R}^n \to {\mathbb R}$ is differentiable almost everywhere, our initial problem is to take function $g(Y)=\nabla f(Y)$, defined almost everywhere, and define it in some ``canonical way'' in the remaining points. The most obvious way to define a function $g$ at a point $Y$ is ``by continuity'': if $g$ is undefined at $Y$ but the limit
\begin{equation}\label{eq:limcont}
\lim\limits_{X \to Y} g(X)
\end{equation}
exists, we may define $g(Y):=\lim\limits_{X \to Y} g(X)$, to make function $g$ continuous at $Y$.

For the limit \eqref{eq:limcont} to be well-defined, we need $g$ to be defined in all points (as opposite to almost all points) in some neighbourhood of $Y$. For example, if $g$ in undefined at all points with rational coordinates, then the limit \eqref{eq:limcont} is not well-defined for all $Y$. For functions defined almost everywhere, the notion of approximate continuity can be used instead of continuity. 
A measurable function $g:E\to {\mathbb R}^m$ defined on $E\subset {\mathbb R}^n$ is called approximately continuous at $Y$ if there is a measurable set $F\subset E$ which has density $1$ at $Y$ and such that 
\begin{equation}\label{eq:apprcont}
g(Y) = \lim\limits_{X\in F, X\to Y} g(X).
\end{equation}
The limit in the right-hand side of \eqref{eq:apprcont} may exist even if $g$ in undefined in some measure $0$ set in any neighbourhood of $Y$.

If $g$ is essentially bounded in some neighbourhood of $Y$, then $g$ is approximately continuous at $Y$ if and only if $Y$ is a Lebesgue point of $g$, that is,
\begin{equation}\label{eq:lebpoint}
\lim\limits_{\epsilon \to 0^+}\frac{1}{|B_\epsilon(Y)|}\int_{B_\epsilon(Y)} |g(X)-g(Y)| dX = 0,
\end{equation}
see Section 1.7.2 of \cite{evans2015measure}.
The Lebesgue differentiation theorem states that, for any integrable function $g:{\mathbb R}^n\to{\mathbb R}^m$, almost every $Y \in {\mathbb R}^n$ is a Lebesgue point of $g$, see Section 1.7.1 of \cite{evans2015measure}. In particular, \eqref{eq:lebpoint} implies that
\begin{equation}\label{eq:lebcont}
g(Y)=\lim\limits_{\epsilon \to 0^+}\frac{1}{|B_\epsilon(Y)|}\int_{B_\epsilon(Y)} g(X)\,dX
\end{equation}

This motivates the following definition
\begin{definition}\label{def:lebcont}
We say that function $g:{\mathbb R}^n \to {\mathbb R}^m$ is Lebesgue continuous at $Y \in {\mathbb R}^n$ if the limit 
\begin{equation}\label{eq:leblim}
\lim\limits_{\epsilon \to 0^+}\frac{1}{|B_\epsilon(Y)|}\int_{B_\epsilon(Y)} g(X) dX
\end{equation}
exists and is equal to $g(Y)$.
\end{definition}

If $g$ is continuous at $Y$, then it is also approximately continuous at $Y$, $Y$ is a Lebesgue point of $g$, and, in turn, $g$ is Lebesgue continuous at $Y$. The Lebesgue continuity at $Y$, however, does not imply that $Y$ is a Lebesgue point. For example, take $n=m=1$ and function $g(x)=\text{sign}(x)$ (that is, $g(x)=1$, $g(x)=0$ and $g(x)=-1$  for $x>0$, $x=0$, and $x<0$, respectively). Then $g$ is not continuous at $0$, not approximately continuous, and $0$ is not a Lebesgue point of $g$, because 
$$
\lim\limits_{\epsilon \to 0^+}\frac{1}{2\epsilon}\int\limits_{-\epsilon}^\epsilon |g(X)-g(0)| dX = 1 \neq 0.
$$
However, $g$ is Lebesgue continuous at $0$, because 
$$
\lim\limits_{\epsilon \to 0^+}\frac{1}{2\epsilon}\int\limits_{-\epsilon}^\epsilon g(X) dX = 0 = g(0).
$$

Lebesgue continuity also has the following probabilistic interpretation. Assume that we ``measure'' $Y$ with a random error $\epsilon Z$, where $Z$ is uniformly distributed\footnote{In fact, the error $Z$ may equivalently be normally distributed, provided that $E[|g(Y+\epsilon Z)|]$ is finite for some $\epsilon>0$. This follows from the rotation invariance property of the multivariate standard normal distribution.} in a unit ball. Then the limit \eqref{eq:leblim} is $\lim\limits_{\epsilon\to 0^+}E[g(Y+\epsilon Z)]$, where $E[\cdot]$ denotes the expected value. Hence, $g$ is Lebesgue continuous at $Y$ if and only if the expected value of $g(Y+\epsilon Z)$ converges to $g(Y)$ when the magnitude $\epsilon$ of (rotation invariant) error $\epsilon Z$ goes to $0$. This allows us to estimate $g(Y)$ by computing $g$ at points of the form $Y+\epsilon Z$, and taking the average.

Theorem \ref{thm:extgrad2} implies that for $g(X)=\nabla f(X)$ with $f: {\mathbb R}^n \to {\mathbb R}$ convex, the limit \eqref{eq:leblim} exists for all $Y \in {\mathbb R}^n$. Hence, it is natural to define the gradient in points where it does not exist by Lebesgue continuity. This is exactly what we did. In fact, \emph{the extended gradient $G_Y(f)$, treated as as function of $Y$ for a fixed $f$, is the unique map from ${\mathbb R}^n$ to ${\mathbb R}^n$, which (i) coincides with $\triangledown f(Y)$ whenever the latter exists, and (ii) is Lebesgue continuous for all $Y \in {\mathbb R}^n$}.

Multiplying both sides of \eqref{eqn:extgrad2} by $Z^T$ for any $Z \in {\mathbb R}^n$, we get
\begin{equation}\label{eqn:dirderav}
Z^T G_Y(f) = \lim\limits_{\epsilon \to 0^+}\frac{1}{|B_\epsilon(Y)|}\int_{B_\epsilon(Y)} Z^T\nabla f(X) dX =  \lim\limits_{\epsilon \to 0^+}\frac{1}{|B_\epsilon(Y)|}\int_{B_\epsilon(Y)} \phi_{f,X}(Z) dX,
\end{equation}
that is, $Z^T G_Y(f)$ is the average value of the directional derivative of $f$ in direction $Z$ in a small ball around $Y$. This fact provides another characterization of $G_Y(f)$.

Next we compare $G_Y(f)$ with the existing concepts of generalized derivative in the literature. In fact, subdifferential set $\partial f(Y)$ itself can be considered as a generalization of a gradient which coincides with it at points of differentiability but is well-defined for all convex functions $f$ at all points $Y$. The obvious difference with $G_Y(f)$ is that $\partial f(Y)$ is, in general, set-valued. The same is true for other generalisations of gradient such as Clarke's generalized gradient \cite{clarke2008nonsmooth} or mollified gradient \cite{ermoliev1995minimization}: these gradients may exist for some non-convex and even discontinuous functions, but they are, in general, set-valued. In contrast, $G_Y(f)$ is a single vector in ${\mathbb R}^n$ for any fixed $f$ and $Y$.

Another generalised derivative which exists for some non-differentiable functions is the weak derivative, whose definition is motivated by integration by parts formula. 
If $u \in [a,b]\to{\mathbb R}$ is a function which is not necessarily differentiable but is known to be integrable, then its weak derivative is an integrable function $v \in [a,b]\to{\mathbb R}$ such that the equality
$$
\int\limits_a^b u(t) \phi'(t)dt = -\int\limits_a^b v(t) \phi(t) dt
$$
holds for all infinitely differentiable functions $\phi$ such that $\phi(a)=\phi (b)=0$. This definition can be generalised to higher dimensions. The weak derivative is known to be unique, in the sense that if $v_1$ and $v_2$ are both weak derivatives of $u$ then $v_1(x)=v_2(x)$ for almost all $x$. However, the weak derivative is not helpful in determining the derivative of $u$ at any \emph{specific} point $y$. For example, let $u(t)=|t|$ on $[-1,1]$ and $v(t)$ is such that $v(t)=-1$ and $v(t)=1$ on $[-1,0)$ and $(0,1]$, respectively. Then $v(t)$ is the weak derivative of $u(t)$ for any value of $v(0)$. In contrast, formula \eqref{eqn:extgrad2} allows us to uniquely determine that the extended gradient of $|t|$ at $t=0$ must be $0$.

One more generalised derivative is the approximate derivative related to the notion of approximate continuity \eqref{eq:apprcont}. A measurable function $g:E\to {\mathbb R}^m$ defined on a measurable set $E\subset {\mathbb R}^n$ having Lebesgue density 1 is called approximately differentiable at $Y \in E$ if there is a measurable set $F\subset E$ which has density $1$ at $Y$ and such that $g$ restricted to $F$ is classically differentiable at $Y$. Many important classes of functions, such as functions of bounded variation, are approximately differentiable at almost all points, see Theorem 6.4 in \cite{evans2015measure}. However, there is no known canonical way to define the approximate derivative for the \emph{all} points, even for simple functions like $u(t)=|t|$.

By \eqref{eqn:subdifrepr}, defining extended gradient is equivalent to selecting a point from the subdifferential set $\partial f(Y)$, which is a non-empty convex compact set in ${\mathbb R}^n$. The extended gradient \eqref{eqn:extgrad}-\eqref{eqn:extgrad2} corresponds to the selection of Steiner point of $\partial f(Y)$. One may ask what if we select a different special point from $\partial f(Y)$, for example, the center of gravity (centroid). For any measurable set $T \subset {\mathbb R}^n$, its centroid is given by
\begin{equation}\label{eq:centermass}
c(T)=\frac{\int\nolimits_{{\mathbb R}^n}x I_T(x) dx}{\int\nolimits_{{\mathbb R}^n}I_T(x) dx},
\end{equation} 
where $I_T(x)$ is the characteristic function of $T$. We could then define the centroid-based extended gradient as 
\begin{equation}\label{eq:centroidbased}
G^c_Y(f) = c(\partial f(Y)).
\end{equation} 
However, it is known that, in general, center of gravity is not preserved under set addition: we may have $c(T_1)+c(T_2) \neq c(T_1+T_2)$ for measurable sets $T_1, T_2 \subset \mathbb{R}^n$. Hence, the centroid-based extended gradient $G^c_Y(f)$ is not additive. Also, consider an isosceles triangle with vertex coordinates $(0,1)$, $(-\epsilon,0)$, and $(\epsilon,0)$, where $\epsilon>0$ is small. The centroid of this triangle is $(0,1/3)$, for any $\epsilon>0$. However, when $\epsilon$ converges to $0$, the triangle degenerates to the interval with endpoints $(0,1)$ and $(0,0)$, and the centroid jumps to $(0,1/2)$. Now consider a family of functions $f_\epsilon$ whose $\partial f_\epsilon(Y)$ at some point $Y$ is this triangle, and a function $f_0$ whose $\partial f_0(Y)$ is the limiting interval. For small $\epsilon$, functions $f_\epsilon$ look almost identical to $f_0$ locally at $Y$, but the centroid-based extended gradients \eqref{eq:centroidbased} of $f_\epsilon$ and $f_0$ at $Y$ are completely different. By contrast, the Steiner point of the same triangle has coordinates $(0,1/2+h(\epsilon))$, where $\lim\limits_{\epsilon\to 0}h(\epsilon)=0$, hence the Steiner point converges to $(0,1/2)$ as $\epsilon\to 0$ smoothly, with no jumps. Hence, the extended gradient \eqref{eqn:extgrad}-\eqref{eqn:extgrad2} of $f_\epsilon$ converges smoothly to the extended gradient of $f_0$. Property (G3) ensures that this is true in general. 

Of course, we may select a different point from $\partial f(Y)$ instead of centroid, but Theorem \ref{thm:extgrad} implies that the Steiner point is the unique choice for which the corresponding extended gradient satisfies the natural properties (G1)-(G4).

Can we define extended gradient for convex functions in an infinite dimensional Banach space $B$? By \eqref{eqn:subdifrepr}, this is equivalent to developing a methodology for selecting a point from the subdifferential set $\partial f(Y)$, which is a non-empty convex weakly compact subset of $B$. Such methodology was developed in \cite{lim1981center}, and can be used to construct a version of extended gradient for convex functions in $B$ which satisfies some useful properties, including (G2), (G4) and a weaker version of (G3), see \cite{Grechuk2015}. However, the additivity (G1) fails, and, in fact, this is unavoidable: it is known that there is no analogue of Steiner point of a convex body in an infinite-dimensional space, see \cite{vitale1985steiner}. This implies that any attempt to extend the notion of extended gradient to infinite dimensional spaces will result in a failure of at least one of properties (G1)-(G4).  

In addition to lacking additivity and axiomatic foundation, the construction in \cite{lim1981center} uses ordinals and transfinite induction, hence the resulting extended gradient is in general difficult to compute. The explicit formula \eqref{eqn:extgrad2} from this paper is not directly applicable to infinite dimensional spaces, because the integration in it is with respect to the Lebesgue measure, and there is no useful analogue of Lebesgue measure on an infinite dimensional Banach space. Specifically, every translation-invariant measure that is not identically zero assigns infinite measure to all open subsets of such space, see \cite{hunt1992prevalence}.

\section{An application to capital allocation and risk sharing}\label{sec:capital}

Let $(\Omega, \Sigma, {\mathbb P})$ be  a probability space, where $\Omega$ denotes the designated space of future states $\omega$, $\Sigma$ is a $\sigma$-algebra of sets in $\Omega$, and ${\mathbb P}$ is a probability measure on $(\Omega, \Sigma)$. A random variable (r.v.) is any measurable function from $\Omega$ to $\mathbb R$. Let $L^0(\Omega)$ be a vector space of all random variables on $\Omega$, and let $V$ be a subspace of $L^0(\Omega)$.

\subsection{Capital allocation}

Let r.v. $Y \in V$ represent a (random) profit of some portfolio, consisting of $m$ sub-portfolios, that is, $Y=\sum_{i=1}^m X_i$. 
Let $\rho:V \to {\mathbb R}$ be a function such that $\rho(X)$ represents the \emph{risk} associated with any $X \in V$.
The \emph{capital allocation problem} is the problem of distributing risk capital $\rho(Y)$ among sub-portfolios, that is, assigning to sub-portfolio $i$ its \emph{risk contribution} $k_i$ such that $\sum_{i=1}^m k_i=\rho(Y)$. 

One solution to the capital allocation problem is to assign to each sub-portfolio $i$ the risk contribution
\begin{equation}\label{eq:eiler}
k_i = \lim\limits_{h\to 0}\frac{\rho(Y+hX_i)-\rho(Y)}{h},
\end{equation}
provided that the limit exists. 
The right-hand side of \eqref{eq:eiler} is called Gateaux derivative of $\rho$ at $Y$ in the direction $X_i$. If $\rho$ is Fr\'{e}chet differentiable at $Y$, then Gateaux derivative is linear in $X_i$, and 
$$
\sum_{i=1}^m k_i = \lim\limits_{h\to 0}\frac{\rho(Y+h \sum_{i=1}^m X_i)-\rho(Y)}{h} = \lim\limits_{h\to 0}\frac{\rho(Y+h Y)-\rho(Y)}{h}.
$$
Hence, we have $\sum_{i=1}^m k_i=\rho(Y)$ if the risk measure $\rho$ satisfies the following property: 
\begin{itemize}
\item[($\rho 1$)] (positive homogeneity): $\rho(\lambda X) = \lambda \rho(X)$ for all $\lambda \geq 0$ and $X \in V$.
\end{itemize}
In this case, the risk contributions defined in \eqref{eq:eiler} are called Euler contributions, and the corresponding allocation method is called the Euler principle, or the Gradient Principle. As noted in \cite{bauer2013capital}, different authors studied the capital allocation problems from different perspectives, but many of them end up in the same allocation method - the Euler principle - if we ignore differences in notation and presentation. Many authors, including \cite{denault2001coherent}, \cite{kalkbrener2005}, and \cite{tasche2007capital} suggest a set of axioms an allocation method should satisfy and then deduce that if there exists an allocation method satisfying their axioms, it is unique and in fact coincides with the Euler principle. For example, \cite{kalkbrener2005} postulated that the risk contribution $k_i$ should depend only on $X_i$ and $Y$, but not on the decomposition of $Y-X_i$ among the rest of sub-portfolios. In this case, capital allocation is a function $\Lambda_\rho: V \times V \to {\mathbb R}$ such that $\Lambda_{\rho}(Y,Y)=\rho(Y), \, \forall\,Y \in V$. 
\cite{kalkbrener2005} also postulated that $\Lambda_\rho$ should satisfy the following properties 
\begin{itemize}
\item[(i)] $X \mapsto \Lambda_{\rho}(X,Y)$ is a linear function;
\item[(ii)] $\Lambda_{\rho}(X,Y)\leq \rho(X)$ for all $X,Y \in V$, and
\item[(iii)] $\Lambda_{\rho}(X,Y)$ is continuous in $Y$, in the sense that $\lim\limits_{\epsilon\to 0} \Lambda_{\rho}(X,Y+\epsilon X) = \Lambda_{\rho}(X,Y)$ for all $X\in V$.
\end{itemize}
Condition (i) guarantees that $\sum\nolimits_{i=1}^m k_i=\rho(Y)$, where $k_i=\Lambda_{\rho}(X_i,Y)$. Condition (ii) is called diversification: the contributions of single assets will never exceed the total risks. Condition (iii) guarantees that the small changes in $Y$ can not change the risk contribution significantly.
See \cite{kalkbrener2005} for further discussion and justification of these axioms.
Theorem 4.2 in \cite{kalkbrener2005} implies that for a capital allocation $\Lambda_{\rho}$ to satisfy (i) and (ii), the risk measure $\rho$ must satisfy ($\rho 1$) and  
\begin{itemize}
\item[($\rho 2$)] (sub-additivity): $\rho(X+Y) \leq \rho(X) + \rho(Y)$ for all $X,Y \in V$.
\end{itemize}
Theorem 4.3 in \cite{kalkbrener2005} guarantees that, for any positively homogeneous and sub-additive $\rho$, a capital allocation $\Lambda_{\rho}$ satisfying (i)-(iii) exists if and only if $\rho$ is differentiable at $Y$, and in this case it is unique and is given by \eqref{eq:eiler}.

Conditions ($\rho 1$) and ($\rho 2$) together ensure that $\rho$ is convex and therefore globally Lipschitz continuous. Hence, it is differentiable almost everywhere on $V$. Based on this, one may hope that the Euler principle \eqref{eq:eiler} works in all practically important situations, and the (measure zero) cases when it fails have only theoretical meaning but no practical value. 
However, Example 3 in \cite{Grechuk2015} demonstrates that if the initial portfolio $Y$ is not ``arbitrary'' but is a result of natural risk-minimization policy, then it may happen that $Y$ is ``forced'' to belong to exactly the (measure zero) set in which $\rho$ is \emph{not} differentiable. Hence, the problem of identifying a \emph{unique} capital allocation scheme which works for \emph{all} $Y$ has not only theoretical, but also practical importance.

To resolve this issue at least partially, \cite{cherny2011two} replaced the continuity property (iii) by the property of law-invariance, that is, $\Lambda_{\rho}(X,Y)$ should depend only on the joint law of $X$ and $Y$. They proved that the capital allocation satisfying (i), (ii), and their new axiom exists and is unique for a class of risk measures they call weighted VaR. However, such capital allocation may not exist or not be unique for risk measures outside of this class.

\cite{Grechuk2015} suggests a way to extend the Euler principle for an arbitrary risk measure satisfying ($\rho 1$) and ($\rho 2$) and arbitrary $X_i$ and $Y$, but the construction involves ordinals and transfinite induction and is therefore impractical in general. In addition, it lacks axiomatic foundation.

We now show how the technique developed in Section \ref{sec:gradient} can be used to develop a system of natural axioms defining a \emph{unique} capital allocation scheme. We assume that the space $V$ is finite-dimensional, and is identified with ${\mathbb R}^n$ for some $n$. This assumption holds in (at least) two important special cases.

The first case is when the underlying probability space $(\Omega, \Sigma, {\mathbb P})$ is finite. In many applications, the underlying distributions are unknown and are estimated based on finite samples of historical data. In these cases, discrete random variables defined on a finite probability space may be an appropriate model.
If $\Omega=(\omega_1, \dots, \omega_n)$ is finite, there is an obvious bijection between random variables $X:\Omega\to{\mathbb R}$ and vectors $u=(u_1,\dots,u_n)\in {\mathbb R}^n$: let $u(X)=(X(\omega_1), \dots, X(\omega_n))$ be the vector consisting on all values of $X$. Conversely, for every $u\in {\mathbb R}^n$ denote  $X_u$ the random variable such that $X_u(\omega_i)=u_i$, $i=1,\dots,n$. The risk measure $\rho$ can then be treated as a function $r:{\mathbb R}^n\to {\mathbb R}$ defined by
$$
r(u) = \rho(X_u).
$$ 

The second case is when probability space $(\Omega, \Sigma, {\mathbb P})$ is arbitrary, and $V$ is the space of all random variables $X$ representable in the form $X=\sum\limits_{i=1}^n u_i X_i$, where $X_i$ represent sub-portfolios, and $u=u(X)=(u_1, \dots, u_n)$ is a vector of real coefficients.
In this case, we can define function $r:{\mathbb R}^n \to {\mathbb R}$ as
$$
r(u) = \rho\left(\sum\limits_{i=1}^n u_i X_i\right),
$$ 
and the Euler contributions \eqref{eq:eiler} can equivalently be written as
\begin{equation}\label{eq:eiler2}
k_i = \frac{\partial r}{\partial u_i}(1,\dots,1),
\end{equation}
provided that $r(u)$ is differentiable at $(1,\dots,1)$.
This set-up has been considered by many authors including \cite{denault2001coherent} and \cite{tasche2007capital}.

Now we suggest to relax the continuity property (iii) of \cite{kalkbrener2005} to a weaker version 
\begin{itemize}
\item[(iv)] $Y \mapsto \Lambda_{\rho}(X,Y)$ is a Lebesgue continuous function on $V$,
\end{itemize}
where Lebesgue continuity is defined in Definition \ref{def:lebcont}. Then we have the following result.

\begin{theorem}\label{th:allocation}
Let $\rho$ be a risk measure satisfying ($\rho 1$) and ($\rho 2$). Then there exists a unique capital allocation scheme satisfying properties (i), (ii), and (iv). This scheme is given by 
\begin{equation}\label{eq:allocation}
\Lambda^*_\rho(X,Y) = G_{u(Y)}(r) \cdot u(X),
\end{equation}
where $G_{u(Y)}(r)$ is the extended gradient of $r$ at $u(Y)$ defined in \eqref{eqn:extgrad}-\eqref{eqn:extgrad2}, and $\cdot$ in the usual scalar product in ${\mathbb R}^n$. Hence, the risk contributions are
\begin{equation}\label{eq:allocation2}
k_i = \Lambda^*_\rho\Big(X_i,\sum X_i\Big) = G_{u(\sum X_i)}(r) \cdot u(X_i).
\end{equation}
\end{theorem} 
\begin{proof}
Let $\Lambda_\rho(X,Y)$ be a capital allocation scheme satisfying (i), (ii), and (iv). By linearity (i), we have 
$$
\Lambda_\rho(X,Y) = f_r(Y) \cdot u(X),
$$
for some vector $f_r(Y) \in {\mathbb R}^n$ which depends on $r$ and $Y$. If $r$ is differentiable at $u(Y)$ with gradient $\nabla r(u(Y))$, then, by Theorem 4.3 in \cite{kalkbrener2005},  
$$
\Lambda_\rho(X,Y) = \lim\limits_{h\to 0}\frac{\rho(Y+hX)-\rho(Y)}{h} = \lim\limits_{h\to 0}\frac{r(u(Y)+hu(X))-r(u(Y))}{h}
$$
$$
= \nabla r(u(Y)) \cdot u(X) = G_{u(Y)}(r) \cdot u(X) = \Lambda^*_\rho(X,Y).
$$
By Lebesgue continuity (iv), equality $\Lambda_\rho(X,Y)=\Lambda^*_\rho(X,Y)$ extends to all $Y \in V$. 
\end{proof}

We end this section by saying a few words about property (iv). Linearity (i) implies that $\Lambda_{\rho}(X,Y)$ is continuous in $X$, which means that small changes of sub-portfolio $X_i$ has a limited effect on $k_i=\Lambda_{\rho}(X_i,Y)$, provided that the total portfolio $Y$ is fixed. It would be desirable and natural to also require continuity in $Y$ (property (iii)), which would imply that small changes in $Y$ also have a limited effect on the risk
capital of its subportfolios. However, the capital allocation satisfying (i), (ii), and (iii) may not exist. Failure of (iii) implies that we may have $\Lambda_{\rho}(X_i,Y+Z)$ significantly different from $\Lambda_{\rho}(X_i,Y)$ even if $Z$ is a small fluctuation. Imagine a (somewhat extreme) scenario when, for all small $Z$, we would have $\Lambda_{\rho}(X_i,Y+Z) \geq a > b > \Lambda_{\rho}(X_i,Y)$ for some $a>b$. This allocation scheme looks strange and unfair: if the risk contribution of sub-portfolio $i$ is at least $a$ for every small fluctuation of $Y$, why it is less than $b$ at $Y$? In this scenario, the intuition is that $k_i=\Lambda_{\rho}(X_i,Y)$ is ``unfairly small''. 
Property (iv) is a natural relaxation of continuity (iii), which precludes this and similar scenarios. It guarantees that $\Lambda_{\rho}(X,Y)$ is equal to the average value of $\Lambda_{\rho}(X,Z)$ when $Z$ takes values in a small ball with center $Y$. Hence, $k_i=\Lambda_{\rho}(X_i,Y)$ is not ``unfairly small'' and not ``unfairly large'', but is just right.

\subsection{Risk sharing}\label{subsec:risk_sharing}

The capital allocation problem arises in different contexts, for example, in the context of risk sharing with cash invariant risk measures. Hence, our solution of the capital allocation problem automatically implies a solution to this problem as well.

Instead of reviewing an extensive literature on risk sharing, we refer the reader to survey \cite{aase2002perspectives} and proceed to problem formulation.
Assume that there are $m$ agents, indexed by $I=\{1,2,\dots, m\}$. Each agent $i \in I$ has an initial endowment $Y_i\in V$, and an associated risk measure $\rho_i:V\to{\mathbb R}$. The agents aim to redistribute the total endowment $Y=\sum_{i=1}^m Y_i$ among themselves to reduce their risk. Agent $i \in I$ receives the part $X_i \in V$ of the total endowment such that $\sum_{i=1}^m X_i=Y$;  the vector $\alloc{X}=(X_1,X_2,\dots,X_m)$ is called the \emph{risk allocation}. A risk allocation $\alloc{X}$ is called \emph{Pareto optimal} if there is no risk allocation $\alloc{Z}=(Z_1,Z_2,\dots,Z_m)$ with $\rho_i(Z_i)\leq\rho_i(X_i), \, i\in I$, with at least one inequality being strict. If the vector $\alloc{Y}=(Y_1,Y_2,\dots,Y_m)$ of all initial endowments is not Pareto optimal, it is beneficial for all agents to switch to a Pareto optimal one. However, there are typically many Pareto optimal allocations, so how to choose a ``fair'' one among them?

If risk measures $\rho_i$ are cash-invariant, that is, $\rho_i(X+C)=\rho_i(X)-C$ for every $X\in V$ and every constant 
$C$,
then risk allocation $\alloc{X}=(X_1,X_2,\dots,X_m)$ is Pareto optimal if and only if it minimizes the total risk $\sum_{i=1}^m \rho_i(X_i)$ over all possible risk allocations. 
If, moreover, all $\rho_i$ are also positively homogeneous and sub-additive, then so is the functional
\begin{equation}\label{eq:coal}
\rho^*(Y) = \inf\limits_{\alloc{X}:\sum_{i=1}^m X_i = Y} \sum_{i=1}^m \rho_i(X_i), 
\end{equation}  
mapping the total endowment $Y$ to the corresponding total risk.
In this case, if a risk allocation $\alloc{X}=(X_1,X_2,\dots,X_m)$ is Pareto optimal, then all allocations $\alloc{X}'=(X_1+C_1,X_2+C_2,\dots,X_m+C_m)$, where $C_i$ are constants such that $\sum_{i=1}^m C_i =0$, are Pareto optimal as well. Moreover, it is easy to see that in fact all Pareto optimal allocations, up to equivalence\footnote{Allocations $\alloc{X}=(X_1,\dots,X_m)$ and $\alloc{Z}=(Z_1,\dots,Z_m)$ are equivalent if $\rho_i(X_i)=\rho_i(Z_i)$ for all $i=1,\ldots,m$}, can be obtained in this way. Hence, the problem of choosing a ``fair'' allocation from the set of Pareto optimal allocation reduces to the ``fair'' selection of constants $C_i$.

Now, each agent $i$ starts with an initial endowment $Y_i$ and ends up with the final endowment $X_i$. Hence, the extra risk taken by the agent is $X_i-Y_i$. The more risk agent agrees to take, the higher premium $C_i$ she should get. Hence, the problem of determining $C_i$ is exactly the problem of determining risk contribution of $X_i-Y_i$ to $Y$, where the aggregate risk measure is $\rho^*$. If, similar to the previous section, we assume that $C_i$ should depend only on $X_i-Y_i$ and $Y$, write $C_i=\Lambda_{\rho^*}(X_i-Y_i,Y)$, and postulate the properties (i), (ii), and (iv), then constants $C_i$ are determined uniquely by Theorem \ref{th:allocation}. Specifically, 
$$
C_i = \Lambda^*_{\rho^*}(X_i-Y_i,Y), \quad i=1,\dots, m, 
$$
where $\Lambda^*$ is defined in \eqref{eq:allocation}.

\section{Mean-deviation portfolio optimization}\label{sec:portf}

\subsection{Finitely generated deviation measures}\label{sec:2}

Assume that the probability space $\Omega$ is finite with $N = |\Omega|$ and $\prob(\omega) > 0$ for any $\omega \in \Omega$. A finite probability space $\Omega$ will be called \emph{uniform}, if $\prob[\omega_1]=\dots=\prob[\omega_N]=\frac{1}{N}$. 

Let $\br^{(i)}$, $i=1,\dots,n$, be random variables denoting the rates of return of financial instruments. We assume that there exists also a risk-free instrument with a constant rate of return $R^{(0)}=:r_0$. Following \cite{rockafellar2006c}, we also assume that 
\begin{assumption}
\item[(M)] any portfolio $\bX=\sum\nolimits_{i=1}^n x_i \br^{(i)}$ is a non-constant random variable for any non-zero $\bx=(x_1, \dots, x_n) \in \er^n$.\label{ass:non-constant}
\end{assumption}
 \cite{rockafellar2006c} formulated portfolio optimization problem as follows
\begin{equation}\label{eqn:portf1}
\min\limits_{(x_0, x_1, \dots, x_n)} \cD\left(\sum\limits_{i=0}^n x_i \br^{(i)}\right), \quad \text{s.t.}\,\, \sum\limits_{i=0}^n x_i = 1, \,\,\,\, \sum\limits_{i=0}^n x_i \ee[\br^{(i)}] \geq r_0+\Delta,
\end{equation} 
where $\Delta>0$ and $\cD$ is a general deviation measure, that is, a functional $\cD:{\cal L}^2(\Omega)\to[0;\infty]$ satisfying:
\begin{assumption}
\setlength{\itemsep}{0pt}
\item[(D1)] $\cD(\bX)=0$ for constant $X$, but $\cD(\bX)>0$ otherwise \emph{(non-negativity)},\label{ass:D1}
\item[(D2)] $\cD(\lambda \bX) = \lambda \cD(\bX)$ for all $X$ and all $\lambda > 0$ \emph{(positive homogeneity)}, \label{ass:D2}
\item[(D3)] $\cD(\bX + Y)\leq \cD(\bX) + \cD(Y)$ for all $\bX$ and $Y$ \emph{(subadditivity)},
\item[(D4)] set $\{\bX\in{\cal L}^2(\Omega)\big| \cD(\bX)\leq C\}$ is closed for all $C<\infty$ \emph{(lower semicontinuity)}.
\end{assumption}
With centered rates of return $\hat{\br}^{(i)}=\br^{(i)}-E[\br^{(i)}]$, $i=1,\dots,n$, and $\mu_i=E[\br^{(i)}] - r_0$,  $i=1,\dots,n$, problem (\ref{eqn:portf1}) can be reformulated as
\begin{equation}
\label{eqn:MV1}
\min_{x \in \er^n} \cD(\hat R^T x),\qquad \text{s.t. } \mu^T x \ge \Delta,
\end{equation} 
where $\hat \bR = (\hat \br^{(1)}, \ldots, \hat \br^{(n)})^T$, $\bx = (x_1, \ldots, x_n)^T$, and $\mu = (\mu_1, \ldots, \mu_n)^T$. We chose to use a distinct notation $\cD$ to indicate this particular choice of the risk measure $\rho$ in \eqref{eqn:intro} as this will be the default setting for the rest of the paper. 

By \citet[Theorem 1]{rockafellar2006b}, every deviation measure $\cD$ can be represented in the form
\begin{equation}\label{eqn:devenvelopes}
\cD(X)=\ee X + \sup_{Q \in \cQ}E[-XQ],
\end{equation}
where $\cQ \subset \L^2(\Omega)$ is called \emph{risk envelope} and can be recovered from $\cD$ by
\begin{equation}\label{eqn:devenvelopes2}
\cQ=\big\{\,Q \in {\cal L}^2(\Omega) \;\big|\;\; \ee[X(1-Q)]\leq \cD(X) \;\; \forall X \in \L^2(\Omega)\big\}.
\end{equation}
Moreover, the set $\cQ$ is closed and convex in $\L^2(\Omega)$. Elements $Q \in \cQ$ for which supremum in (\ref{eqn:devenvelopes}) is attained are called \emph{risk identifiers} of $X$. The set of all risk identifiers of $X$ is denoted $\cQ(X)$.

A deviation measure $\cD$ is \emph{finite}, that is, $\cD(X)<\infty, \, \forall \, X$ if and only if the corresponding $\cQ$ is bounded. In this case, $\cQ(X)$ is non-empty for every $X \in \L^2(\Omega)$, and, due to closeness and convexity of $\cQ$ and linearity of $Q \mapsto \ee[-XQ]$, every set $\cQ(X)$ must contain at least one extreme point of $\cQ$. Therefore, $\sup_{Q \in \cQ} \ee[-XQ]= \max_{Q \in \cQ^e} \ee[-XQ]$, where $\cQ^e$ is the set of all extreme points of $\cQ$. 
In fact, a bounded closed convex $\cQ$ is the closed convex hull of $\cQ^e$, see Theorem 2 in \cite{phelps1974dentability}\footnote{
Because $\L^2(\Omega)$ is a reflexive Banach space, it has the Radon-Nikodym property, and 
Theorem 2 in \cite{phelps1974dentability} applies.}.
Of particular importance to this paper will be the set of such risk measures for which the set $\cQ^e$ is finite:
\begin{definition}
A finite deviation measure $\cD$ is called \emph{finitely generated} if the set $\cQ^e$ of all extreme points of $\cQ$ is finite. We will call elements of this set \emph{extreme risk generators}.
\end{definition}
In other words, $\cD$ is finitely generated if and only if $\cQ$ is a convex hull of a finite number of points.

\begin{ex}
For standard deviation, $\sigma(X) =||X-E[X]||_2$, the risk envelope is given by \citet[Example 1]{rockafellar2006c} 
$$
\cQ=\big\{\,Q \;\big|\; E[Q]=1, \;\; \sigma(Q) \leq 1\big\},
$$
and, for $N>2$, has infinitely many extreme points, hence $\sigma$ is \emph{not} finitely generated. 
\end{ex}

\begin{ex}\label{ex:2}
For mean absolute deviation, ${\rm MAD}(X) =E[|X-E[X]|]$, the risk envelope is given by \citet[Example 2]{rockafellar2006c} 
$$
\cQ=\big\{\,Q \;\big|\; E[Q]=1, \;\; \sup Q-\inf Q \leq 2\big\},
$$
which is a convex polytope in ${\mathbb R}^N$ with a finite number of vertices. Hence ${\rm MAD}$ is finitely generated. In fact, extreme points $\cQ^e$ can be explicitly written as 
$$
\cQ^e = \big\{\,Q = 1 + \ee[Z] - Z \;\big|\; \exists S \subset \{1,2,\dots, N\}: \, Z_i=1, \,\, i\in S; \,\, Z_i=-1, \; i\not\in S\big\},
$$
which for a uniform probability on $\Omega$ simplifies to
$$
\cQ^e = \big\{\, \bx \in{\mathbb R}^N\;\big|\; \exists S \subset \{1,2,\dots, N\}: \, x_i=\frac{2|S|}{N}-1, \,\, i\in S; \,\, x_i=\frac{2|S|}{N}+1, \; i\not\in S\big\},
$$
where the subset $S$ is taken non-empty and proper. Hence, $|\cQ^e|=2^N-2$.
\end{ex}

\begin{ex}\label{ex:cvar}
For CVaR-deviation
\begin{equation}\label{eqn:cvardevdef}
{\rm CVaR}_\alpha^\Delta(X)\equiv E[X]-\frac{1}{\alpha}\int\nolimits_{0}^{\alpha}q_X(\beta)\,d\beta,
\end{equation}
the risk envelope is given by \citet[Example 4]{rockafellar2006c} 
$$
\cQ_\alpha=\big\{\,Q \;\big|\; E[Q]=1, \;\; 0\leq Q \leq \alpha^{-1}\big\}.
$$
The linearity of constraints imply that ${\rm CVaR}_\alpha^\Delta$ is finitely generated. In particular, if the probability is uniform over $\Omega$ and $\alpha=\frac{k}{N}$ for some integer $1\leq k<N$, extreme points $\cQ^e$ are
$$
\cQ^e = \big\{\,\bx \in{\mathbb R}^N\;\big|\; \exists S \subset \{1,2,\dots, N\}: \, |S|=k, \;\; x_i=\frac{N}{k}, \,\, i\in S; \,\, x_i=0, \; i\not\in S\big\}.
$$
This implies that $|\cQ^e|=\frac{N!}{k!(N-k)!}$.
\end{ex}

\begin{lemma}\label{lem:fingen}
Let $\cD_1, \cD_2, \dots, \cD_m$ be finitely generated deviation measures. Then functionals
\begin{itemize}
\item[(a)] $\cD(X) = \sum_{i=1}^m \lambda_i \cD_i(X)$, with $\lambda_i>0,\, i=1\dots,m$;
\item[(b)] $\cD(X) = \max\{\cD_1(X), \dots, \cD_m(X)\}$
\end{itemize}
are also finitely generated deviation measures.
\end{lemma}
\begin{proof}
Proof follows from \citet[Proposition 4]{rockafellar2006b}, and from the fact that if sets $\cQ_1, \cQ_2, \dots, \cQ_m$ are all convex hulls of a finite number of points, then so are the sets: $\lambda_1 \cQ_1 + \dots + \lambda_m Q_m$; the convex hull of $\cQ_1 \cup \dots \cup \cQ_m$; and $\{Q\,|\,Q=(1-\lambda)+\lambda Q_i\,\, \text{for some}\,\, Q_i \in \cQ_i\}$, $\lambda > 0$, $i=1, \ldots, m$.
\end{proof}

\begin{ex}\label{ex:mixedcvar}
Mixed CVaR-deviation 
\begin{equation}\label{eqn:mixedcvar}
\cvar^\Delta_{\lambda}(X)=\int_0^1 \cvar^\Delta_{\alpha}(X)\,\lambda(d\alpha),
\end{equation}
where $\lambda$ is a probability measure on $(0,1)$, is also finitely generated. Indeed, because the probability space is finite, mixed CVaR-deviation \eqref{eqn:mixedcvar} can be written as a finite mixture of CVaR-deviations
$$
\cvar^\Delta_{\lambda}(X) = \sum_{i=1}^m \lambda_i {\rm CVaR}_{\alpha_i}^\Delta(X),
$$
where $\alpha_i \in (0,1)$, $\lambda_i>0$, $i=1,\dots,m$, and $\sum_{i=1}^m \lambda_i =1$, which is a finitely generated deviation measure due to Example \ref{ex:cvar} and Lemma \ref{lem:fingen}(a).
\end{ex}

\subsection{Optimal portfolios and active portfolio risk generators}\label{sec:4}

We make the following standing assumptions:
\begin{assumption}
\item[(A)] The deviation measure $\cD$ is finitely generated.\label{ass:1}
\item[(B)] $\Delta > 0$ and $\mu \ne \0$.\label{ass:2}
\end{assumption}
The latter assumption implies the following properties of the optimal solution to \eqref{eqn:MV1}.
\begin{lemma}\label{lem:positive}
The optimal objective value in \eqref{eqn:MV1} is positive and in optimum the constraint is binding: $\mu^T x = \Delta$.
\end{lemma}
\begin{proof}
By Theorem 1 in \cite{rockafellar2006c}, there is an optimal solution $x^*$. By assumption \ref{ass:2}, $x = \0$ does not satisfy the constraint on the expected return, and so $x^* \ne \0$. Due to assumption \ref{ass:non-constant}, we conclude that $\hat R^T x^*$ is random and hence $\cD(\hat R^T x^*) > 0$. For the second part of the statement, assume that $\mu^T x^* > \Delta$. Therefore, there is $\eta < 1$ such that $\mu^T (\eta x^*) \ge \Delta$ and we have $\cD(\hat R^T (\eta x^*)) = \eta \cD(\hat R^T x^*) < \cD(\hat R^T x^*)$, a contradiction.
\end{proof}

Since $\cD$ is finitely generated, the deviation measure of a centered return of portfolio $x \in \er^n$ can be expressed as a maximum of a finite number of terms:
\begin{equation}\label{eqn:extreme_representation}
\cD(\hat R^T x) = \max_{Q \in \cQ^e} \ee[-\hat R^T x\, Q].
\end{equation}
As the number of extreme risk generators for $\cD$ is finite, they can be enumerated:
$
\cQ^e = \{ Q_1, \ldots, Q_{M'} \}.
$
Define $\tilde D_i = \ee[-\hat R Q_i]$, $i = 1, \ldots, M'$.  It follows from \eqref{eqn:extreme_representation} that the set of $\tilde D_i$'s is sufficient to evaluate $\cD(\hat R^T x)$ for a portfolio $x$:
\begin{equation}\label{eqn:risk_gen_representation1}
\cD(\hat R^T x) = \max_{i =1, \ldots, M'} \tilde D_i^T x.
\end{equation}
It may happen that $\tilde D_i = \tilde D_j$ for some $i \ne j$; for example, $\hat R$ may be constant on a number of elementary events in $\Omega$. It may also happen that $\tilde D_i$ is not an extreme point of $\conv \{ \tilde D_1, \ldots, \tilde D_{M'} \}$, but a subset of $\tilde D_i$'s forms all extreme points of this set \citet[Theorem IV.19.3]{rockafellar1970}. For the convenience of future arguments, we choose only those vectors $\tilde D_i$ that are extreme points.
\begin{definition}\label{def:portfolio_risk_gen}
Extreme points of $\conv \{ \tilde D_1, \ldots, \tilde D_{M'} \}$ are denoted by $D_i$, $i=1, \ldots, M$, and called \emph{portfolio risk generators}. 
\end{definition}
\begin{remark}
Portfolio risk generators are generators (in the sense of \citet[Section 19]{rockafellar1970}) of the polyhedral set $\{ \ee [ -\hat R Q ] \,|\, Q \in \cQ \}$, see the proof of Theorem 19.3 in \cite{rockafellar1970}.
\end{remark}
By Definition \ref{def:portfolio_risk_gen} and  \eqref{eqn:risk_gen_representation1}, it is easy to see that for any portfolio $x$:
\begin{equation}\label{eqn:risk_gen_representation}
\cD(\hat R^T x) = \max_{i =1, \ldots, M} D_i^T x.
\end{equation}
\begin{definition}
Those $D_i$ that realize the maximum in \eqref{eqn:risk_gen_representation} are called \emph{active portfolio risk generators} for the portfolio $x$.
\end{definition}

The following lemma shows that the set of portfolio risk generators is sufficiently rich to span the whole space $\er^n$.
\begin{lemma}\label{lem:full_gen}
$\lin(D_1, \ldots, D_M) = \er^n$.
\end{lemma}
\begin{proof}
Assume the opposite and take any non-zero vector $x$ in the orthogonal complement of $\lin(D_1, \ldots, D_M)$. Then $\cD(\hat R^T x) = 0$. However, $\hat R^T x$ is non-constant by assumption \ref{ass:non-constant}, so its deviation measure should be strictly positive by \ref{ass:D1}. A contradiction.
\end{proof}

The representation \eqref{eqn:risk_gen_representation} of the deviation measure of a portfolio $x$ enables an equivalent formulation of optimization problem \eqref{eqn:MV1} as a linear program:
\begin{equation}\label{eqn:lin}
\begin{aligned}
&\text{minimize } A,\\
&\text{subject to: } A \ge D_i^T x, \quad i=1, \ldots, M,\\
&\phantom{\text{subject to: }} \mu^T x \ge \Delta,\\
&\phantom{\text{subject to: }}(A, x) \in \er \times \er^{n}.
\end{aligned}
\end{equation}
The solution $(A^*, x^*)$ is related to \eqref{eqn:MV1} as follows: $x^*$ is the optimal portfolio while $A^* = \cD(\hat R^T x^*)$. 

\begin{theorem}\label{thm:mu}
The linear program \eqref{eqn:lin} as well as the optimization problem \eqref{eqn:MV1} have the following properties:
\begin{enumerate}
\item The set of optimal portfolios $\X^*$ is a bounded polyhedral subset of $\er^n$. The set of solutions to \eqref{eqn:lin} is of the form $\{A^*\} \times \X^*$ for some $A^* > 0$.
\item If the solution is not unique then $\mu$ is a linear combination of at most $n-1$ portfolio risk generators.
\item If the solution is unique, then the set of active portfolio risk generators spans the whole space $\er^n$, i.e., there are $n$ linearly independent active portfolio risk generators.
\end{enumerate}
\end{theorem}
\begin{proof}
\eqref{eqn:lin} is a linear program, so the set of solutions is polyhedral. The mapping $x \mapsto \cD(\hat R^T x)$ is convex, hence also continuous. Denote by $d$ its minimum on the sphere $\{ x \in \er^n\,|\, \|x\| = 1 \}$. This minimum is strictly positive due to assumptions \ref{ass:non-constant} and \ref{ass:D1}. Employing further assumption \ref{ass:D2} gives that $\{x \in \er^n\,|\, \cD(\hat R^T x) \le A \}$ is bounded for any $A > 0$; indeed, it is contained in the ball with radius $A/d$. Hence, the set of solutions $\X'$ to \eqref{eqn:lin} is a bounded polyhedral set. It is expressed by convex combinations of its extreme points at which the objective function is optimal. In each such extreme point the coordinate $A$ is identical, so $\X' = \{A^*\} \times \X^*$ for some $A^* > 0$; the positivity of $A^*$ follows from Lemma \ref{lem:positive}.

If $\X'$ is a single point, then it is an extreme point of the feasible set. Since the constraint $\mu^T x \ge \Delta$ is active (see Lemma \ref{lem:positive}), \cite[Theorem 2.2]{Bertsimas1997} implies that there are $n$ indices $i_1, \ldots, i_n$ such that $A = D_{i_j}^T x$, $j=1, \ldots, M$, and vectors $(D_{i_j})_{j=1}^n$ are linearly independent, hence generate $\mathbb{R}^n$. 

The proof of assertion 2 uses the dual of problem \eqref{eqn:lin}:
\begin{equation}\label{eqn:dual} 
\begin{aligned}
&\text{maximize } q \Delta,\\
&\text{subject to: }  \sum_{i=1}^M p_i  D_i - q\mu = 0,\qquad \sum_{i=1}^M p_i =1,\\
&\phantom{\text{subject to: }} q \ge 0,\ p_i \ge 0, \quad i=1,\dots, M.
\end{aligned}
\end{equation}
By the strong duality, $q\Delta = A^*$ and we know $A^* > 0$, hence $q > 0$. \citet[Theorem 4.5]{Bertsimas1997} implies that the dual variables corresponding to inactive constraints are zero. Denote by $i_1, \ldots, i_k$ the active constraints involving portfolio risk generators. Then  the first constraint in the above dual problem \eqref{eqn:dual} reads:
\begin{equation}\label{mu_active}
  \mu =  \frac{1}{q} \sum_{j=1}^k p_{i_j} D_{i_j}.
\end{equation}

Assume now that the solution is not unique, i.e., $\X'$ contains at least two extreme points and therefore a line connecting them. Fix an internal point of that line $(A^*, x^*)$. 
Since $(A^*, x^*)$ is not an extreme point of $\X'$, the linear space spanned by active portfolio risk generators $D_{i_j}$, $j= 1,\dots, k$, has dimension not larger than $n-1$ (there is at least one portfolio risk generator which is active at an extreme point of $\X'$ and does not belong to $\text{lin} \{ D_{i_1}, \ldots, D_{i_k} \}$). This proves assertion 2 of the theorem.
\end{proof}

\begin{corollary}\label{cor:unique}
There is a finite number of hyperplanes (of dimensions from $1$ to $n-1$) such that:  $\mu$ belongs to one of them if and only if a solution to \eqref{eqn:MV1} is not unique. Therefore, the set of $\mu$ for which the portfolio optimization problem has a unique solution has a full Lebesgue measure.
\end{corollary}
\begin{proof}
By Theorem \ref{thm:mu}, non-uniqueness of solutions coincides with $\mu$ being a linear combination of at most $n-1$ portfolio risk generators, i.e., belongs to a linear space spanned by at most $n-1$ vectors in $\er^n$. This is a hyperplane of dimension at most $n-1$, so it has a Lebesgue measure $0$. There is a finite number of ways to choose up to $n-1$ vectors from the set of $M$ vectors, so the number of such hyperplanes is finite. A finite sum of sets of Lebesgue measure zero has the measure zero. Its complement has therefore a full measure.
\end{proof}

A practical consequence of the above theorem and corollary is that there is a unique optimal portfolio in \eqref{eqn:MV1} unless $\mu$ is specially chosen to match the distribution of returns $\hat R$ and the risk measure. In the following section we will show that the uniqueness of a solution, which implies multiple active portfolio risk generators, leads to problems with optimal cooperative investment. We will also show that there are natural settings when $\mu$ happens to be on one of the hyperplanes mentioned in the corollary.

Consider now a portfolio optimization problem with no shortsales of risky assets. This corresponds to the linear program \eqref{eqn:lin} with additional constraints $x_i \ge 0$, $i=1, \ldots, n$. The following lemma shows that the non-uniqueness of portfolio risk generators holds here as well. 
\begin{lemma}
If the solution $x^*$ to the portfolio optimization problem with no shortsales of risky assets is unique, then there are at least $k$ active portfolio risk generators, where $k$ is the number of non-zero coordinates of $x^*$.
\end{lemma}
\begin{proof}
As in Lemma \ref{lem:positive}, we show that $\mu^T x^* = \Delta$. The assertion follows from the fact that $x^*$ is an extreme point of the feasible set.
\end{proof}

\section{Cooperative investment}\label{sec:7}

\subsection{Theoretical framework}

The general problem of cooperative investment can be formulated as follows, see \cite{grechuk2015synergy}. Let ${\cal F}\subset {\cal L}^2(\Omega)$ be a feasible set, representing rates of return from feasible investment opportunities on the market without a riskless asset:
$$
{\cal F}=\Big\{X\,\Big|\,X=\sum_{i=1}^n \br^{(i)} x_i,\;\; \sum_{i=1}^n x_i = 1\Big\}.
$$
An individual portfolio optimization problem for agent $i$, $i=1,\dots, m$, is 
\begin{equation}\label{eq:portf_ind}
\max_{X\in{\cal F}}\, U_i(X),
\end{equation}
where $U_i:{\cal L}^2(\Omega) \to [-\infty, \infty)$ is the utility function of agent $i$. 
If the unit of capital is invested, then rate of return $X \in {\cal F}$ can also be interpreted as a monetary profit from the investment.
Instead of investing individually, $m$ agents can invest their $m$ units of capital to buy a joint portfolio $X\in m{\cal F}:= \{ mX\,| X \in \mathcal{F}\}$ and distribute it so that agent $i$ receives a share $Y_i$ with $\sum Y_i = X$. An allocation $\alloc{Y}=(Y_1, \dots Y_m)$ is called \emph{feasible} if $\sum Y_i \in m{\cal F}$, and Pareto optimal if there is no feasible allocation $\alloc{Z}=(Z_1, \dots Z_m)$ such that $U_i(Y_i)\leq U_i(Z_i)$ with at least one inequality being strict.

A utility function $U$ is called \emph{cash-invariant} if $U(X+C)=U(X)+C$ for all $X\in{\cal L}^2(\Omega)$ and $C\in{\mathbb R}$. Proposition 2 in \cite{grechuk2015synergy} implies that if all $U_i$, $i=1,\dots,m$, are cash-invariant, and $\alloc{Y}=(Y_1, \dots Y_m)$ is Pareto optimal, then $X^*=\sum Y_i$ solves the optimization problem
\begin{equation}\label{eq:portf_star}
\sup_{X\in m{\cal F}}\,U^*(X),
\end{equation}
where 
\begin{equation}\label{eq:conv_def}
U^*(X)\equiv\sup_{\alloc{Z}\in{\cal A}(X)}\sum_{i=1}^m U_i(Z_i)
\end{equation}
with ${\cal A}(X)=\big\{\alloc{Z} = (Z_1, \ldots, Z_m):\ \sum_{i=1}^m Z_i=X, \ Z_i \in \L^2(\Omega)\big\}$. Furthermore, if $\alloc{Y}=(Y_1, \dots Y_m)$ is \emph{any} Pareto optimal allocation, then \emph{all} Pareto optimal allocations are given by 
\begin{equation}\label{eq:plusc}
(Y_1+C_1,\ldots,Y_m+C_m),
\end{equation} 
where $C_1,\dots C_m$ are constants with $\sum_{i=1}^m C_i=0$. Hence, the coalition should (i) solve the portfolio optimization problem \eqref{eq:portf_star} to find an optimal portfolio $X^*$ for the whole group; (ii) find \emph{any} Pareto optimal way $\alloc{Y}=(Y_1, \dots Y_m)$ to distribute $X$ among group members, and finally (iii) agree on constants $C_1,\dots C_m$ in \eqref{eq:plusc} to select a \emph{specific} Pareto-optimal allocation among the ones available.

We consider investors employing the following utility functions:
\begin{equation}\label{eqn:Ui}
U_i(X)=\ee[X]-\cD_i(X),
\end{equation}
for some deviation measures $\cD_i$, $i=1,\dots,m$. These utility functions are cash-invariant and the above theory applies. $U^*$ in \eqref{eq:conv_def} is given by $U^*(X)=\ee[X]-\cD^*(X)$, where
\begin{equation}\label{eq:devcoal}
\cD^*(X)\equiv\inf_{\alloc{Z}\in{\cal A}(X)}\sum_{i=1}^m \cD_i(Z_i).
\end{equation}
\begin{remark}\label{rem:3.1}
Commonly, an investor's optimization criterion is given by
\[
U_i(X)=\ee[X]-\gamma_i\cD_i(X),
\]
where $\gamma_i > 0$ is the investor's risk aversion. However, $\gamma_i \cD_i$ is a deviation measure whenever $\cD_i$ is, so the expression \eqref{eqn:Ui} covers this example.
\end{remark}

In this model, a possible approach to (iii) is to select constants $C_i$ in \eqref{eq:plusc} such that
\begin{equation}\label{eq:fair}
E[Q^*(Y_1 + C_1)] = \dots = E[Q^*(Y_m + C_m)]
\end{equation}
where $Q^*$ is the extreme risk identifier in portfolio optimization problem \eqref{eq:portf_star} with $U^*(X)=\ee[X]-\cD^*(X)$.
The intuition is that elements $Q$ of risk envelope represents probability scenarios, $Q^*$ represents the ``critical'' worst-case scenario for the coalition, and \eqref{eq:fair} states that the investors should receive the same profit under the critical scenario. See  \citet[Section 3]{grechuk2015synergy} for further justification of \eqref{eq:fair} in the model with risk-free asset.
Because a concave function is differentiable almost everywhere, one may expect that $\partial U^*(X^*)$ is ``typically'' a singleton, in which case the extreme risk identifier $Q^*$ is unique, and this approach leads to the unique selection of a ``fair'' Pareto optimal allocation in \eqref{eq:plusc}. Below we show, however, that this intuition may be wrong.

\begin{lemma}\label{lem:devcoal}
Let $\cD_i$ be deviation measures with risk envelopes $\cQ_i$, $i=1,\dots,m$. Then $\cD^*$ is a deviation measure with the risk envelope $\cQ^*=\cQ_1 \cap \dots \cap \cQ_m$. In particular, if all $\cD_i$ are finitely generated, then so is $\cD^*$. 
\end{lemma}
\begin{proof}
Proposition 3 in \cite{rockafellar2006b} implies that $\cQ_1, \dots, \cQ_m$ are closed, convex subsets of the closed hyperplane $H = \{Q | EQ=1 \}$ in ${\cal L}^2(\Omega)$ such that constant $1$ is in their quasi-interior relative to $H$. Because $\cQ_1, \dots, \cQ_m$ have a common point in their relative interiors, \citet[Corollary 16.4.1]{rockafellar1970} implies that $\cD^*$ can be represented in the form \eqref{eqn:devenvelopes} with $\cQ^*=\cQ_1 \cap \dots \cap \cQ_m$. Because $\cQ^*$ is also closed, convex subset of $H$ with constant $1$ in quasi-interior relative to $H$, this implies that $\cD^*$ is a deviation measure. Because intersection of polygons is a polygon, $\cD^*$ is finitely generated provided that all $\cD_i$ are.
\end{proof}
  
 \begin{theorem}\label{thm:coop}
Assume that investors' utility functions are of the form $U_i (X) = \ee[X] - \cD_i(X)$ with deviation measures $\cD_i$ finitely generated and none of the portfolio risk generators $D^*_i$ for $\cD^*$ is equal to $\mu = \ee[R]$ or $(D^*_i-\mu)$ is parallel to $\mathbf{1} := (1, \ldots, 1)^T$. Then any solution $X^* = R^T x^*$ to \eqref{eq:portf_star} has at least two extreme risk identifiers. 
 \end{theorem}
\begin{proof}
We follow ideas from the proof of Theorem \ref{thm:mu}. Denote by $(D^*_i)_{i=1}^M$ the portfolio risk generators for the deviation measure $\cD^*$ and let ${\hat D}^*_i =  D^*_i - \ee[R]$. Then \eqref{eq:portf_star} is equivalent to the 
following linear problem
\begin{equation} 
\begin{aligned}
&\text{minimize } A,\\
&\text{subject to: } A \ge x^T{\hat D}^*_i, \quad i=1, \ldots, M,\\
&\phantom{\text{subject to: }} x^T \mathbf{1} = 1, \qquad (A, x) \in \er \times \er^{n}.
\end{aligned}
\end{equation}
Since $x^*$ is a solution to this program (not necessarily unique), its dual also has a solution \citep[Theorem 4.4]{Bertsimas1997}:
\begin{equation} 
\begin{aligned}
&\text{maximize } q,\\
&\text{subject to: }  \sum_{k=1}^M p_k {\hat D}^*_k - q\mathbf{1} = 0,\qquad \sum_{k=1}^M p_k =1,\\
&\phantom{\text{subject to: }} p_k \ge 0,\ k=1,\dots, M,\quad q \in \mathbb{R}.  
\end{aligned}
\end{equation}
If the optimal solution $q \ne 0$, then the middle equation together with the assumption that none of $\hat D^*_j$'s is parallel to $\mathbf{1}$ implies that there must be at least two $p_k$'s strictly positive. \citet[Theorem 4.5]{Bertsimas1997} states that the corresponding constraints in the primal problem are active, i.e., their respective portfolio risk generators are active for $X^*$. When $q = 0$, the assumption that none of $\hat D^*_j$'s is zero imply again that at least two $p_k$'s must be non-zero.
\end{proof}

Theorem \ref{thm:coop} implies that there are at least two linearly independent active portfolio risk generators and there are multiple fair Pareto-optimal solutions to the cooperative investment problem.

\begin{remark}
The portfolio optimization problem \eqref{eq:portf_star} with utility functions \eqref{eqn:Ui} may not have a solution, i.e., an optimal value may be attained asymptotically on a diverging sequence of portfolios. This happens, for example, when there is $x$ such that $x^T \mathbf{1} = 0$ and $\mu^T x - \cD^*(x^T R) > 0$, i.e., when the projection of $\mu$ on $\mathbf{1}^\bot := \{ y \in \mathbb{R}^d:\ y^T \mathbf{1} = 0 \}$ is not contained in the convex envelope of projections  of portfolio risk generators $(D^*_i)_{i=1}^M$ on $\mathbf{1}^\bot$.
\end{remark}

\begin{remark}
The issue with non-existence of solution to the optimization problem of this section 
\begin{equation}\label{eqn:gamma_cD}
\sup_{x:\ x^T \mathbf{1} = 1} x^T \mu - \cD(x^T R) 
\end{equation}
extends to optimization with the risk measured by a coherent risk measure $\rho$ (another criterion popular in the literature)
\[
\sup_{x:\ x^T \mathbf{1} = 1} x^T \mu - \gamma \rho(x^T R)
\]
with the risk aversion $\gamma > 0$. Indeed, using that $\rho(X) = \cD(X) - E[X]$ for some deviation measure $\cD$, the above problem is equivalent to
\[
\sup_{x:\ x^T \mathbf{1} = 1} x^T \mu - \gamma^* \cD(x^T R)
\]
with $\gamma^* = \gamma/(1+\gamma)$, which, by Remark \ref{rem:3.1}, is of the form \eqref{eqn:gamma_cD}.
\end{remark}

\subsection{Explicit example}

Cash-or-nothing binary option $O$ returns some fixed amount of cash $C(O)$ if it expires in-the-money but nothing otherwise. Assume that there are two such options $A$ and $B$ which expire in-the-money if $P>C_1$ and $P>C_2$, respectively, where $P$ is the (random) price of (the same) underlying asset, and $C_1<C_2$ are constants. Assume that options are offered for the same price $p$ with $C(A)=2p$ and $C(B)=8p$. Each agent can invest a unit of capital into $A$ and $B$, precisely $1-t$ into $A$ and $t$ into $B$, to get profit $-(1-t)-t=-1$; $(1-t)-t=1-2t$; or $(1-t)+7t=1+6t$ depending on the relation of the price $P$ with respect to $C_1$ and $C_2$. We assume that two agents think that these three opportunities are equally probable.

For agent $1$ with $U_1(X) = E[X] - CVaR^\Delta_{\frac{2}{3}}(X) = -CVaR_{\frac{2}{3}}(X)$, an optimal individual investment can be found from the linear program
$$
\max_{a_1,t} a_1, \quad \text{s.t. }\, X=(-1,1-2t,1+6t), \quad E[Q\,X] \geq a_1, \,\forall\, Q\in \cQ^1,
$$
where $\cQ^1 = \left\{\left(\frac{3}{2},\frac{3}{2},0\right),\left(\frac{3}{2},0,\frac{3}{2}\right),\left(0,\frac{3}{2},\frac{3}{2}\right)\right\} = \left\{\text{Perm}\left(\frac{3}{2},\frac{3}{2},0\right)\right\}$,
resulting in the optimum $t=0$, $X=(-1,1,1)$, and the optimal value $u_1^*=0$.

Similarly, for agent $2$ with $U_2(X) = E[X]-\frac{1}{2}MAD(X)$, the linear program
$$
\max_{a_2,t} a_2, \quad \text{s.t.}\, X=(-1,1-2t,1+6t), \quad E[Q\,X] \geq a_2, \, \forall\,Q\in \cQ^2,
$$
where
$\cQ^2 =  \left\{\text{Perm}\left(\frac{5}{3},\frac{2}{3},\frac{2}{3}\right), \text{Perm}\left(\frac{4}{3},\frac{4}{3},\frac{1}{3}\right)\right\}$,
returns $t=\frac{1}{5}$, with the optimal value $u_2^*=\frac{1}{15}$.

The cooperative investment corresponds to the linear program
$$
\max_{a_1,a_2,Y_1,Y_2,t} a_1+a_2, \ \text{s.t. } Y_1+Y_2=2(-1,1-2t,1+6t), \, E[Q Y_j] \geq a_j, \, \forall\,Q\in \cQ^j, \, j=1,2, 
$$
that is, we are simultaneously looking for optimal portfolio ($t$), and an optimal way to share it ($Y_1,Y_2$) to maximize the sum of agents utilities. The optimal $t$ is $t=\frac{1}{5}$, with $Y_1+Y_2=\left(-2,\frac{6}{5}, \frac{22}{5}\right)$, and optimal value is $u^*=\frac{2}{15} > u_1^* + u_2^*$. 
The simplex method returns a solution $Y_1=(\frac{2}{15},\frac{2}{15},\frac{2}{15})$, $Y_2=(-\frac{32}{15},\frac{16}{15},\frac{64}{15})$, with $u_1(Y_1)=\frac{2}{15}$ and $u_2(Y_2)=0$, which is obviously unfair. 
Because the utilities are cash invariant, any solution in the form 
$Y'_1=Y_1+C$, $Y'_2=Y_2-C$  
is Pareto-optimal, and the question is how to select a ``fair'' $C$. 

To this end, we compute the utility of a coalition $U^*(X)$ as
\begin{equation}\label{eq:utcoal}
U^*(X) = \min_{Q\in \cQ^*}\,E[QX],
\end{equation}
where $\cQ^*$ can be found as (the vertices of) intersection of convex hulls of $\cQ^1$ and $\cQ^2$. In our case, $\cQ^* = \left\{\text{Perm}\left(\frac{3}{2},1,\frac{1}{2}\right), \text{Perm}\left(\frac{4}{3},\frac{4}{3},\frac{1}{3}\right)\right\}$.
The optimal portfolio $X^*=\left(-2,\frac{6}{5}, \frac{22}{5}\right)$ is a solution to the optimization problem
\begin{equation}\label{eq:porft}
\max U^*(X), \quad \text{s.t. }\, X=2(-1,1-2t,1+6t).
\end{equation}
Now, let $Q^*$ be the minimizer in \eqref{eq:utcoal} for $X^*$. Then, according to \eqref{eq:fair}, the fair $C$ should be selected such that
\begin{equation}\label{eq:cond_for_c}
E[Q^* (Y_1+C)]=E[Q^* (Y_2-C)].
\end{equation}  
An intuition is that the investors should get the same profit under the critical scenario $Q^*$.  The problem is that, for $X^*=\left(-2,\frac{6}{5}, \frac{22}{5}\right)$, the minimizer $Q^*$ in \eqref{eq:utcoal} in not unique! Indeed, $E[Q X^*]=\frac{2}{15}$ for $Q=\left(\frac{3}{2},1,\frac{1}{2}\right)$, and also for $Q=\left(\frac{4}{3},\frac{4}{3},\frac{1}{3}\right)$. This is not a coincidence as we have shown in Theorem \ref{thm:coop}. While set of random variables $X$ with non-unique risk identifier has measure $0$, the optimal portfolio in \eqref{eq:porft} is guaranteed to belong to this set. Consequently, the cooperative investment does not have a unique solution in the case of finitely generated deviation measures.

In our example, the set of minimizers in \eqref{eq:utcoal} is the whole line segment with endpoints $\left(\frac{3}{2},1,\frac{1}{2}\right)$ and $\left(\frac{4}{3},\frac{4}{3},\frac{1}{3}\right)$. Consequently, there are infinitely many ``fair'' choices of $C$.

\section{Inverse portfolio problem}\label{sec:5}

Following \cite{palczewski2016}, let us formulate a problem inverse to (\ref{eqn:MV1}) as follows. Assume that we know a solution $x^M=(x_1^M, \dots, x_n^M) \ne \0$ to (\ref{eqn:MV1}), together with centered rates of return $\hat R$, deviation measure $\cD$, and $\Delta_M > 0$  the expected excess return of the portfolio $x^M$.  Can we then ``recover'' $\mu_i$, the expected excess returns of individual instruments? Are they determined uniquely? We will give a positive answer to the first question and discuss a dichotomy faced by the second: if the solution of the inverse problem is unique then the forward problem with the computed $\mu$ has multiple solutions, while if the forward problem has a unique solution then there are many $\mu$'s solving the inverse problem.

\subsection{An explicit formula using risk generators}\label{sec:explicit_formula}

Assume that $\Delta_M > 0$. Necessarily, $x^M \ne \0$. Theorem 4 in \cite{rockafellar2006c} states that, the portfolio $x^M$ is a solution to \eqref{eqn:MV1} \emph{if and only if} there is a risk identifier $Q^*$ for the random variable $\hat R^T x^M$ such that
\begin{equation}\label{eqn:mu_inv}
\mu = \frac{\Delta_M}{\cD(\hat R^T x^M)} \ee[-\hat R Q^*] =  \frac{\Delta_M}{(x^M)^T\ee [-\hat R Q^*]} \ee[-\hat R Q^*].
\end{equation}
This follows since every finite deviation measure on a discrete probability space is continuous, c.f. \citet[page 518]{rockafellar2006c}.

Let $D_{i_j}$, $j=1, \ldots, k$, be the set of active portfolio risk generators for $x^M$. Then \eqref{eqn:mu_inv} amounts to the existence of weights $\beta_1, \ldots, \beta_k \ge 0$, such that $\sum_{j=1}^k \beta_j = 1$ and
\begin{equation}\label{eqn:mu_inv_finite}
\mu =  \frac{\Delta_M}{\sum_{j=1}^k \beta_j D_{i_j}^T x^M} \sum_{j=1}^k \beta_j D_{i_j}.
\end{equation}
From the above formula we immediately get the following characterization of vectors $\mu$ for which $x^M$ is a solution to \eqref{eqn:MV1}.
\begin{lemma}\label{lem:conv}
The set of solutions $\cM$ to the inverse optimization problem is convex and spanned by points $\delta D_{i_j}$, where $D_{i_j}$, $j=1, \dots, k$,  are active portfolio risk generators for $x^M$ and $\delta = \Delta / \cD(\hat R^T x^M)$:
\[
\cM = \Big\{ \delta \sum_{j=1}^k \beta_j D_{i_j}\,\Big| \, \beta \in [0,1]^k \text{ and } \sum_{j=1}^k \beta_j = 1 \Big\}.
\]
\end{lemma}
\begin{remark}
The conclusions of the above lemma can be immediately deduced from the dual representation \eqref{eqn:dual}  of the portfolio optimization problem. Indeed, we have
\[
\mu = \frac{1}{q} \sum_{i=1}^M p_{i_j}  D_{i_j},
\]
where $q > 0$, $p_{i_j} \ge 0$ and sum up to $1$. Multiplying both sides by $x^M$ yields $q = 1/\delta$.
\end{remark}

Equipped with this characterization of the set $\cM$ we demonstrate the link between the set of solutions of the inverse and forward optimization problems.
\begin{theorem}\label{thm:dichotomy}
$\ $
\begin{enumerate}
\item If $x^M$ is a unique solution to \eqref{eqn:MV1} for some $\mu$, then the set of all solutions $\cM$ to the inverse optimization problem has at least $n + 1$ extreme points. Moreover, all extreme points are of the form $\delta D_{i_j}$, where $\delta > 0$ and $D_{i_j}$ is an active portfolio risk generator for $x^M$.
\item If there is a unique active portfolio risk generator for $x^M$, then the inverse optimization problem has a unique solution $\mu^*$ (the set $\cM$ consists of one point). However, the optimization problem \eqref{eqn:MV1} with $\Delta = \Delta_M$ and $\mu = \mu^*$ has multiple solutions: the set of solutions $\X^*$ is a polyhedron of dimension $n-1$ and has at least $n$ extreme points\footnote{The dimension of a polyhedron $P$ is the maximum number of affinely independent points contained in $P$ minus $1$.}.
\end{enumerate}
\end{theorem}
The proof of the above theorem requires the following simple technical result.
\begin{lemma}\label{lem:ext}
Given $v_i \in \er^n$, $i=1, \ldots, k$, let $\hat n = \rank (v_i,\ i=1, \ldots, k) = \dim (\lin (v_1, \ldots, v_k))$. Then $\cN = \conv (v_1, \ldots, v_k)$ has at least $\hat n + 1$ extreme points and all extreme points are from the set $\{v_1, \ldots, v_k\}$.
\end{lemma}
\begin{proof}
It follows from \citet[Corollary 18.3.1]{rockafellar1970} that all extreme points of $\cN$ are in $\{v_1, \ldots, v_k\}$. It remains to prove that there are at least $\hat n + 1$ extreme points. Assume the opposite: there are only $n' < \hat n + 1$ extreme points $v_{i_1}, \ldots, v_{i_{n'}}$ of $\cN$. Then $\cN \subset A := \lin (v_{i_1}, \ldots, v_{i_{n'}})$ and $\dim(A) \le n' + 1$. However, $A$ is a linear space containing all points $v_1, \ldots, v_k$ so it  also contains $\lin(v_1, \ldots, v_k)$. The latter space has dimension $\hat n + 1$ by assumption, hence a contradiction.
\end{proof}

\begin{proof}[Proof of Theorem \ref{thm:dichotomy}]
From Theorem \ref{thm:mu}, the uniqueness of solutions to \eqref{eqn:MV1} implies that the set of active portfolio risk generators $D_{i_1}, \ldots, D_{i_k}$ spans the whole space $\er^n$, i.e., the dimension of a linear space generated by those vectors is $n$. The conclusions follow from Lemma \ref{lem:ext}.

Assume now that there is a unique active portfolio risk generator. The uniqueness of solution to the inverse optimization problem is clear from formula \eqref{eqn:mu_inv_finite}. Consider the equivalent form \eqref{eqn:lin} for the forward optimization problem. Recall that the set of all solutions to such a linear problem is a convex bounded polyhedral set, a face of a polyhedral set generated by the constraints. The portfolio $x_M$ is a solution for which there are exactly two active constraints: one with the unique portfolio risk generator and one encoding the minimum expected return. This implies that the set of solutions is a polyhedron of dimension $n-1$. By Lemma \ref{lem:ext} it must have at least $n$ extreme points.
\end{proof}

\begin{corollary}
In the case 1 of Theorem \ref{thm:dichotomy}, if $\mu \in \reint \cM$ ($\mu$ is in the relative interior of $\cM$), then the forward optimization problem \eqref{eqn:MV1} has a unique solution for $\Delta = \Delta_M$. 
\end{corollary}
 \begin{proof}
 The implication is equivalent to: solution to \eqref{eqn:MV1} is not unique $\implies \mu \notin \reint \cM$. This follows immediately from assertion 2 of Theorem \ref{thm:mu} and  \citet[Theorem 6.4]{rockafellar1970}.
 \end{proof}

The inverse problem with no shortsales constraint admits more solutions as shown in the following lemma.
\begin{lemma}
Let $x^M$ be a solution to the portfolio optimization problem \eqref{eqn:MV1} with additional constraint of no shortsales of risky assets. The set of solutions to the inverse optimization problem is given by
\begin{equation*}
\cM' = \{ \mu \in \er^n:\ \exists\, m \in \cM \text{ s.t. } \mu_i \le m_i, \text{ and } m_i \ind{x^M_i \ne 0} = \mu_i \ind{x^M_i \ne 0}, \ i = 1, \ldots, n\}
\end{equation*}
In particular, $\cM \subset \cM'$.
\end{lemma}
\begin{proof}
The dual to portfolio optimization problem \eqref{eqn:lin} with non-negativity constraints on portfolio weights of risky assets is given by
\begin{equation}\label{eqn:dual_noshort} 
\begin{aligned}
&\text{maximize } q \Delta_M,\\
&\text{subject to: }  \sum_{j=1}^k p_j  D_{i_j} - q\mu \ge  0,\qquad \sum_{j=1}^k p_j =1,\\
&\phantom{\text{subject to: }} q \ge 0,\ p_j \ge 0, \quad j=1,\dots, k,
\end{aligned}
\end{equation}
where $D_{i_j}$, $j=1, \ldots, k$, are the active portfolio risk identifiers. By the strong duality, $q\Delta_M = \cD(\hat R^T x^M) > 0$, hence $q = \cD(\hat R^T x^M) / \Delta_M > 0$. By complementary slackness conditions, \citet[Theorem 4.5]{Bertsimas1997}, the inequality in the first constraint above becomes equality for those coordinates for which $x^M$ is non-zero. Using the fact that $\mu^T x^M = \Delta_M$ yields the form of $\cM'$. Strong duality implies that for any $\mu \in \cM'$, the portfolio $x^M$ is optimal.
\end{proof}

It is well known that optimal portfolios with shortsales constraints are often poorly diversified, i.e., have many null portfolio weights. It then transpires from the definition of the set $\cM'$ that the coordinates of $\mu$ corresponding to those zero weights are unbounded from below. 
\begin{remark}
If the portfolio with no shortsales $x^M$ is fully diversified, i.e., all coordinates are strictly positive, then assertions of Theorem \ref{thm:dichotomy} apply to the problem with no shortsale constraints for risky assets.
\end{remark}

\subsection{Choice of a single solution to the inverse optimization problem}\label{sec:robustsel}

By Corollary \ref{cor:unique}, the solution to the portfolio optimization problem \eqref{eqn:MV1} is unique unless $\mu$ belongs to a set of Lebesgue measure zero (a union of a finite number of hyperplanes). It is therefore common that the inverse optimization problem has multiple solutions (Theorem \ref{thm:dichotomy}). 

How to choose a unique point from the set of solutions  to the inverse optimization problem? 
In view of \eqref{eqn:mu_inv}, this is equivalent to the choice of a unique risk identifier $Q^*$ or rather a map $f_{\cD}:\L^2(\Omega) \to {\L}^2(\Omega)$ that, for a deviation measure $\cD$, assigns to a random variable $X \in {\L}^2(\Omega)$ one of its risk identifiers. We will call such a map $f_\cD$ a \emph{selector} corresponding to the deviation measure $\cD$. We say that $f_\cD$ is a \emph{robust selector} if it is (i) a selector, and (ii) a Lebesgue continuous map in sense of Definition \ref{def:lebcont}.

\begin{lemma}\label{lem:robsel}
For any finite deviation measure $\cD$ there exists a \emph{unique} robust selector $f_\cD$.
\end{lemma}
\begin{proof}
By \citet[Proposition 1]{rockafellar2006c} we have $\partial \cD(X) = 1 - \cQ(X)$. Hence the existence and uniqueness follow from Theorem \ref{thm:extgrad2}.
\end{proof}

\begin{ex}\label{ex:madsel}
For mean absolute deviation ${\rm MAD}(X)=E[|X-EX|]$, the unique robust selector is given by $f_\cD(X)=1+EZ-Z$, where $Z(\omega)=1$, $Z(\omega)=0$, and $Z(\omega)=-1$, for $X(\omega)>E[X]$, $X(\omega)=E[X]$, and $X(\omega)<E[X]$, respectively.
\end{ex}

An alternative approach for selecting a unique selector is discussed in Appendix. It is based on the principle of law-invariance. Although a law-invariant selector is, in general, not unique, it is natural from the financial and probabilistic point of view and is unique for some important deviation measures such as $\cvar$ and mixed-$\cvar$.

\subsection{Explicit examples}
Let $\Omega=\{\omega_1, \dots, \omega_N\}$ with $\prob(\omega_j)=w_j$, $j=1,\dots, N$, and $\hat{R}_{j}=\hat{R}(\omega_j)$. Consider a given portfolio $x^M$ and denote $X^* = \hat R^T x^M$ and $x_j^*= X^*(\omega_j)$. 
Without loss of generality, we assume that $\{\omega_1, \dots, \omega_N\}$ are ordered in such a way that $x_1^* \leq x_2^* \le \dots \leq x_N^*$. Since $\ee[\hat R] = 0$, we have $\ee[X^*] = 0$ and either $x_1 = \cdots = x_N = 0$ or $x_1 < 0 < x_N$. The former case is impossible for a non-zero portfolio $x^M$ under the assumption \ref{ass:non-constant}, therefore, we will concentrate on the non-trivial latter case of non-zero return $X^*$.
We will examine the inverse portfolio problem for risk measured by MAD and by deviation CVaR.

\subsubsection{Mean absolute deviation}\label{subsec:MAD_example}
Let $k$ be the maximal index such that $x_k^* < 0$ and $m$ be the maximal index such that $x_m^* \leq 0$. It follows from the discussion above that $1 \le k \le m$. The inverse portfolio problem has solutions of the form
\[
\mu = \frac{\Delta_M}{MAD(X^*)} \ee[ - \cQ^* \hat R ],
\]
where $\cQ^*$ is a risk identifier for $X^*$. Recalling the form $\cQ = 1 + \ee[Z] - Z$ of risk generators for MAD, see Example \ref{ex:madsel}, and $\ee[\hat R] = 0$ we get $\ee[ - \cQ^* \hat R ] = \ee[Z \hat R]$. If $k = m$, there is a unique risk identifier given by $Z (\omega_j) = \ind{j > k} - \ind{j \le k}$, $j=1, \ldots, N$. Otherwise, there are $2(m-k)$ extreme risk identifiers corresponding to $Z$'s of the form 
$Z(\omega_j) = \ind{j > m} - \ind{j \le k} + z \ind{j = j^*}$, $j=1, \ldots, N$, for some $k < j^* \le m$ and $z \in \{ -1, 1\}$. Therefore, the set of solutions of the inverse problem is given by
\[
\bigg\{ \sum_{j=m+1}^N w_j \hat R_j - \sum_{j=1}^k w_j \hat R_j + \sum_{j=k+1}^m \lambda_j w_j \hat R^j \ \big| \ \lambda_{k+1}, \ldots, \lambda_m \in [-1,1] \bigg\}.
\]
A robust selector corresponds to taking $\lambda = 0$ (see Example \ref{ex:madsel}).

\begin{ex}
Let $N=3$ and $\prob(\omega_j) = \frac13$, $j=1, 2, 3$. There are two risky assets with centered returns  $\hat R_1 = (-1, -2)^T$, $\hat R_2 = (-1, 1)^T$ and $\hat R_3 = (2, 1)^T$. The solution to the forward portfolio optimization problem with $\mu = (0.4, 0.6)$ and $\Delta_M = 0.5$ is $x_M = (0.5, 0.5)$. Then $X^* = (-1.5, 0, 1.5)$ and $MAD(X^*) = 1$. The set of risk identifiers of $X^*$ is given by $Z=(-1, z, 1)$ with an arbitrary number $z \in [-1,1]$, i.e., $\displaystyle\cQ^* = \Big(2+\frac{z}{3}, 1 - \frac23z, \frac{z}{3}\Big)$. The corresponding set of solutions $\mu$ to the inverse problem takes the form:
$$
\mu = \frac{0.5}{1}\left(\frac13(-1) \hat R_1 + \frac13 z \hat R_2  + \frac13 \hat R_3  \right) = \begin{pmatrix} 0.5 -z/6 \\ 0.5 + z/6\end{pmatrix}, \qquad z \in [-1,1].
$$ 
The unique robust selector suggested in Example \ref{ex:madsel} corresponds to $z=0$, resulting in $\mu = (0.5, 0.5)^T$.
\end{ex}

\subsubsection{Conditional Value at Risk}

For deviation CVaR let $k$ be the maximal index such that $x_k^* < -\var_\alpha(X^*)$ (set $k=0$ is no such index exists) and $m$ be the maximal index such that $x_m^* \leq -\var_\alpha(X^*)$.
Then any risk identifier $Q^*=(q_1, \dots, q_N)$ of $X^*$ satisfies, c.f. \cite{rockafellar2006c},
\begin{equation}\label{eqn:cvarident2}
\begin{cases} 0\leq q_j \leq 1/\alpha, \,\,\,\sum\limits_{j=1}^N w_j q_j=1,\\ q_1=q_2=\dots=q_k=1/\alpha, \\ q_{m+1}=\dots=q_N=0. \end{cases}
\end{equation}
Hence,
\begin{equation}\label{eqn:cvardiscr}
\mu = \frac{\Delta_M}{\cvar_\alpha^\Delta(X^*)}\bigg(\frac{1}{\alpha} \sum\limits_{j=1}^k w_j (-\hat{R}_j) + \sum\limits_{j=k+1}^m w_j q_j(-\hat R_j)\bigg),
\end{equation}
where $q_{k+1}, \dots, q_m$ are arbitrary numbers satisfying linear constraints
$$
\sum\limits_{j=k+1}^m w_j q_j= 1-\frac{1}{\alpha}\sum\limits_{j=1}^k w_j, \quad \text{and} \quad 0\leq q_j \leq 1/\alpha, \quad j=k+1, \dots, m.
$$

If $m=k+1$, the risk identifier in (\ref{eqn:cvarident2}) and $\mu$ in (\ref{eqn:cvardiscr}) are uniquely defined. For $m>k+1$, i.e., $x_{k+1}^* = \cdots = x_m^* = -\var_\alpha (X^*)$, the inverse problem has infinitely many solutions. The 
robust selector corresponds to $q_{k+1}= \dots = q_m$, that is,
$
\mu =\frac{\Delta_M}{\cvar_\alpha^\Delta(X^*)}\big(\frac{1}{\alpha} \sum_{j=1}^k w_j (-\hat R_j) + q\sum_{j=k+1}^m w_j (-\hat{R}_j)\big),
$
where $q=\big(1-\frac{1}{\alpha}\sum_{j=1}^k w_j\big)/\big(\sum_{j=k+1}^m w_j\big)$.  

\begin{ex}\label{ex:10}
Let $\Omega = \{ \omega_1, \omega_2, \omega_3 \}$ with uniform probability $\prob(\omega_j) = 1/3$. There are two risky assets with centered returns $\hat R_1 = (-1, 0)^T$, $\hat R_2 = (0, -1), \hat R_3 = (1,1)$. Fix $\alpha=0.05$. The solution to the forward portfolio optimization problem with $\mu = (1/3, 2/3)$ and $\Delta_M = 0.5$ is $x_M = (0.5, 0.5)$. Then $X^*=(-0.5, -0.5, 1)$,  $-\var_\alpha(X^*)=-0.5$, and  $k=0$, $m=2$. The set of risk identifiers of $X^*$ comprises $Q=(q_1, q_2, 0)$, where $0 \leq q_1, q_2 \leq 20$ and $q_1+q_2=3$. Parameterizing $q_1 = q$ and $q_2=3-q$ for $q \in [0,3]$, we obtain
$$
\mu = \frac{0.5}{1}\left(\frac13 q (-\hat R_1) + \frac13 (3-q) (-\hat R_2)  \right) = \begin{pmatrix} q/3 \\ (3-q)/3\end{pmatrix}, \qquad q \in [0,3].
$$ 
The robust selector is given by $q=1.5$, resulting in $\mu_1 = \mu_2 = 0.5$.
\end{ex}

\section{An application to Black-Litterman portfolio framework}\label{sec:6}
In this section, we apply findings of Section \ref{sec:5} to an extension of the Black-Litterman model of portfolio optimization on markets with discrete distributions of returns. Asset return distributions are commonly approximated with a finite number of scenarios in practical financial applications, see, e.g., \cite{KPU2002, gaivoronski2004, lim2010portfolio, lwin2017mean}. We start with a short presentation of the extension of market-based Black-Litterman model of \cite{meucci2005} to general discrete distributions and deviation measures.\footnote{The reader is referred to \cite{palczewski2016} for a detailed discussion of a parallel extension for continuous distributions.} We demonstrate that the non-uniqueness of solutions to the inverse optimization problem (Section \ref{sec:4}) is commonly observed in this theory and means that the posterior distribution of returns is not unique. The principle of law invariance brings back the well-definiteness of this portfolio theory.

The underlying assumption of the original Black-Litter\-man model \citep{BlackLitterman91} is that the market is in equilibrium in which the mutual fund theorem holds, i.e., all investors hold risky assets in the same proportions. In the general setting of deviation measures, \cite{rockafellar2007} develops an analogous theory and calls the common portfolio of risky assets a master fund. It can be recovered by solving \eqref{eqn:MV1} for a particular choice of $\Delta = \Delta_M$. We assume, as in the original framework, that the market is in equilibrium, so the master fund corresponds to relative market capitalizations of stocks: we will call it a \emph{market portfolio} $x^M$. Further, acting in the spirit of \cite{BlackLitterman91} we assume that the centered equilibrium distribution is known, for example, it is equal to the centered empirical distribution of asset returns. The only parameter of the distribution which is unknown is its location. To recover the latter, we solve an inverse optimization problem: knowing the solution $x^M$ to problem \eqref{eqn:MV1} we find the mean excess return vector $\mu_{eq}$ for a given expected market return $\Delta = \Delta_M$. The distribution $\mu_{eq} + \hat R$ is then called \emph{equilibrium distribution} or prior distribution.

Investor's views are represented by a $m \times n$ `pick matrix' $P$ and a vector $v \in \er^m$. Each row of $P$ specifies combinations of assets and the corresponding entry in $v$ provides a forecasted excess return. The uncertainty (the lack of confidence) in the forecasts is represented by a zero-mean random variable $\varepsilon$ with a continuous distribution with full support on $\er^m$, for example, a normal distribution $N(0, Q)$. 
The resulting Bayesian model is
\begin{align*}
&\text{prior: } \quad R \sim \mu_{eq} + \hat R,\\
&\text{observation: } \quad V|[R=r] \sim Pr + \varepsilon.
\end{align*}
The posterior distribution of future returns $R$ given $V = v$ is concentrated on the same points as the prior distribution but with different probabilities. It can be described by a new probability measure $\mathbb{Q}$ on $\Omega$, i.e., the posterior distribution of asset excess returns is that of $\mu_{eq} + \hat R$ under $\mathbb{Q}$\footnote{The location vector of the posterior distribution is rarely equal to $\mu_{eq}$ due to the reweighing of probabilities in $\mathbb{Q}$ relative to $\prob$}. Following Bayes formula, we set the unnormalized ``density'' of the posterior distribution:
\[
X(\omega) = f_\varepsilon \big(v - P \mu_{eq} - P \hat R(\omega) \big),
\]
where $f_\varepsilon$ is the density of $\varepsilon$. Then $\mathbb{Q}(\omega)/\prob (\omega) = X (\omega) / E_\prob[X]$. The posterior distribution of asset returns is then fed into the optimization problem \eqref{eqn:MV1}.

Assume now that the deviation measure $\cD$ is finitely generated. By Corollary \ref{cor:unique} it should be expected that the market portfolio is a unique solution to \eqref{eqn:MV1}. Consequently, the inverse optimization problem that determines the equilibrium distribution has many solutions (Theorem \ref{thm:dichotomy} and Example \ref{ex:10}) resulting in multitude of posterior distributions and, in effect, multitude of Black-Litterman optimal portfolios. This is obviously unacceptable in a financial context. This non-uniqueness is caused by the existence of many active portfolio risk generators (active risk identifiers for the deviation measure), Lemma \ref{lem:conv}. Selecting the law-invariant active risk identifier, see Section \ref{sec:robustsel} and Example \ref{ex:10}, brings back uniqueness of the solution to the inverse optimization problem and, consequently, the uniqueness of solution to the complete portfolio optimization workflow.

In practice, an investor commonly infers the market portfolio from the market capitalization of assets. Such a portfolio is unlikely to have more than one active portfolio risk generator since optimal portfolios with at least two active portfolio risk generators lie on a finite number of hyperplanes in $\er^n$ (their Lebesgue measure is zero). Hence, the market portfolio solves an \textit{unlikely} portfolio optimization problem for which the set of solutions has dimension $n-1$, see Theorem \ref{thm:dichotomy}. The inverse optimization problem has, however, a unique solution.

\begin{ex}
Consider the setting of Example \ref{ex:10}. Extreme risk identifiers for $\cvar^\Delta_{5\%}$ are $\mathcal{Q}^e = \{ \text{Perm}(3,0,0)\}$. The set of portfolio risk generators consists of $3$ vectors:
\[
D_1 = (1,0)^T, \quad D_2 = (0,1)^T, \quad D_3 = (-1, -1)^T.
\]
Fix a market portfolio $x^M = (0.2, 0.8)^T$ and its return $\Delta_M = 0.4$. The only active portfolio risk generator for $x^M$ is $D_2$. From Lemma \ref{lem:conv}, the inverse optimization problem has a unique solution
$\mu^* = (0, 0.5)$. Consider now the forward optimization problem with expected excess return $\Delta_M$ and mean excess return $\mu^*$:
\[
\min_{x_1, x_2} \ \max \big( x_1; x_2; -x_1-x_2\big), \qquad \text{subject to: } 0.5 x_2 \ge 0.4.
\]
The set of solutions is $\X^* = \big\{ (x_1, 0.8):\ x_1 \in [-1.6, 0.8] \big\}$.  Each solution in $\X^*$ has $\cvar^\Delta_{5\%}$ equal to $0.8$ and the expected excess return of $\Delta_M$. 
\end{ex}

\section{Conclusions}\label{sec:concl}

We have analyzed in depth forward and inverse portfolio optimization problems when asset returns follow a finite number of scenarios and deviation measure is finitely generated (covering popular deviation measures: CVaR, mixed CVaR and MAD). We discovered a dichotomy in the uniqueness of solutions for both problems: the forward and inverse problems cannot be simultaneously uniquely solved (for the same data). Nevertheless, the set of parameters for which the non-uniqueness holds is of measure zero. Although it may seem that the uniqueness problem is practically negligible, we have demonstrated that this is not true in many applications, like capital allocation, cooperative investment, and the generalized Black-Litterman model. In cooperative investment, the non-uniqueness affects a ``fair'' way of distributing profit of joint investment between participating investors: for investors with preferences described by utility functions derived from finitely generated deviation measures, when the coalition's forward optimization problem has a unique solution (which happens on the set of model parameters of full measure), there are many risk identifiers for the optimal wealth which prevents a unique ``fair'' allocation of wealth between investors. For the generalized Black-Litterman model, the inverse optimization problem has multiple solutions resulting in multiple posterior distributions and optimal portfolios. This result is in contrast with the classical Black-Litterman model where the uniqueness holds for both forward and inverse problems.

The above problem of non-uniqueness has been shown to be connected to the fact that a convex function (here a risk or deviation measure) may  not be everywhere differentiable, and, at points of non-differentiability, has a non-unique sub-gradient. This issue has been addressed by introducing the set of axioms, such that, for any convex function and at every point, there is a \emph{unique} sub-gradient satisfying these axioms. This sub-gradient happens to coincide with the Steiner point of the sub-differential set.

In forward optimization problems, if the solution is not unique, we can optimize amongst those solutions according to a secondary objective. For example, if there are many optimal portfolios, we can choose the one which is the ``closest'' to our current portfolio to minimize rebalancing, c.f. \cite{Palczewski2018}. In the inverse optimization, we try to identify values of parameters such that a given specific solution is optimal. If this can be done in several ways, it is unclear how the multi-objective principle described above can help in identifying the unique solution. Also, it may be unclear what secondary objective to choose in such applications. Instead of introducing secondary objectives, Section \ref{sec:gradient} suggests that the set of optimal solution contains exactly one special solution satisfying some highly desirable properties, such as Lebesgue continuity, and recommends to choose this solution. We further show that under some conditions, a robust selector can be characterised as the only law-invariant solution, the property which is natural in applications where distributions are the only observable characteristics of the model.

We have demonstrated applications of our theoretical results, such as Theorem \ref{thm:mu} and Theorem \ref{thm:extgrad}, to various problems in risk analysis and portfolio optimization. While the considered applications are specific in their assumptions, e.g., short selling is allowed, specific risk measures involved, specific constraint structure, etc., the theory itself is very general. Theorems \ref{thm:mu} and \ref{thm:dichotomy} can be extended to a broad class of (parametrized) linear programs, and explain why non-uniqueness issue is ``common'' in direct and inverse linear optimization. Theorem \ref{thm:extgrad} provides a method to assign a unique gradient to any convex function at any point, hence it resolves the non-uniqueness issue in all cases  when it is a consequence of non-differentiability of a convex function.

\bibliographystyle{apalike}
\bibliography{MV_IP}

\begin{thebibliography}{}

\bibitem[Aase, 2002]{aase2002perspectives}
Aase, K.~K. (2002).
\newblock Perspectives of risk sharing.
\newblock {\em Scandinavian Actuarial Journal}, 2002(2):73--128.

\bibitem[Arrow, 1963]{arrow1963}
Arrow, K.~J. (1963).
\newblock Uncertainty and the welfare economics of medical care.
\newblock {\em The American Economic Review}, 53:941--973.

\bibitem[Bauer and Zanjani, 2013]{bauer2013capital}
Bauer, D. and Zanjani, G.~H. (2013).
\newblock Capital allocation and its discontents.
\newblock In {\em Handbook of Insurance}, pages 863--880. Springer.

\bibitem[Bertsimas et~al., 2012]{BGP2012}
Bertsimas, D., Gupta, V., and Paschalidis, I.~C. (2012).
\newblock Inverse optimization: A new perspective on the {B}lack-{L}itterman
  model.
\newblock {\em Operations Research}, 60(6):1389--1403.

\bibitem[Bertsimas and Tsitsiklis, 1997]{Bertsimas1997}
Bertsimas, D. and Tsitsiklis, J. (1997).
\newblock {\em Introduction to linear optimization}.
\newblock Athena Scientific, Belmont, Massachusetts.

\bibitem[Black and Litterman, 1992]{BlackLitterman91}
Black, F. and Litterman, R. (1992).
\newblock Global portfolio optimization.
\newblock {\em Financial Analysts Journal}, 48:28--43.

\bibitem[Borch, 1962]{borch1962}
Borch, K. (1962).
\newblock Equilibrium in a reinsurance market.
\newblock {\em Econometrica}, 30:424--444.

\bibitem[Cherny and Orlov, 2011]{cherny2011two}
Cherny, A. and Orlov, D. (2011).
\newblock On two approaches to coherent risk contribution.
\newblock {\em Mathematical Finance}, 21(3):557--571.

\bibitem[Cherny, 2006]{cherny2006}
Cherny, A.~S. (2006).
\newblock Weighted {V@R} and its properties.
\newblock {\em Finance and Stochastics}, 10(3):367--393.

\bibitem[Clarke et~al., 2008]{clarke2008nonsmooth}
Clarke, F.~H., Ledyaev, Y.~S., Stern, R.~J., and Wolenski, P.~R. (2008).
\newblock {\em Nonsmooth analysis and control theory}, volume 178.
\newblock Springer Science \& Business Media.

\bibitem[Conejo et~al., 2010]{conejo2010decision}
Conejo, A.~J., Carri{\'o}n, M., Morales, J.~M., et~al. (2010).
\newblock {\em Decision making under uncertainty in electricity markets},
  volume~1.
\newblock Springer.
\newblock Chapter 3 is about generating scenarios.

\bibitem[Dana, 2005]{dana2005representation}
Dana, R.-A. (2005).
\newblock A representation result for concave {S}chur concave functions.
\newblock {\em Mathematical Finance}, 15(4):613--634.

\bibitem[Denault, 2001]{denault2001coherent}
Denault, M. (2001).
\newblock Coherent allocation of risk capital.
\newblock {\em Journal of Risk}, 4:1--34.

\bibitem[Dentcheva, 1998]{dentcheva1998}
Dentcheva, D. (1998).
\newblock Differentiable selections and {C}astaing representations of
  multifunctions.
\newblock {\em Journal of Mathematical Analysis and Applications},
  223(2):371--396.

\bibitem[Dhaene et~al., 2012]{dhaene}
Dhaene, J., Tsanakas, A., Valdez, E.~A., and Vanduffel, S. (2012).
\newblock Optimal capital allocation principles.
\newblock {\em Journal of Risk and Insurance}, 79(1):1--28.

\bibitem[Ermoliev et~al., 1995]{ermoliev1995minimization}
Ermoliev, Y.~M., Norkin, V.~I., and Wets, R.~J. (1995).
\newblock The minimization of semicontinuous functions: mollifier subgradients.
\newblock {\em SIAM Journal on Control and Optimization}, 33(1):149--167.

\bibitem[Evans and Gariepy, 2015]{evans2015measure}
Evans, L.~C. and Gariepy, R.~F. (2015).
\newblock {\em Measure theory and fine properties of functions}.
\newblock CRC press.

\bibitem[Fabozzi et~al., 2010]{fabozzi2010robust}
Fabozzi, F.~J., Huang, D., and Zhou, G. (2010).
\newblock Robust portfolios: contributions from operations research and
  finance.
\newblock {\em Annals of operations research}, 176(1):191--220.
\newblock Section 3.5.2 deals with discrete distributions and explains that
  this is a practically important example.

\bibitem[Filipovi{\'c} and Kupper, 2008]{filipovic2008}
Filipovi{\'c}, D. and Kupper, M. (2008).
\newblock Equilibrium prices for monetary utility functions.
\newblock {\em International Journal of Theoretical and Applied Finance},
  11(3):325--343.

\bibitem[F{\"o}llmer and Schied, 2011]{follmer2011stochastic}
F{\"o}llmer, H. and Schied, A. (2011).
\newblock {\em Stochastic Finance: An Introduction in Discrete Time}.
\newblock Walter de Gruyter.

\bibitem[Gaivoronski and Pflug, 2004]{gaivoronski2004}
Gaivoronski, A.~A. and Pflug, G. (2004).
\newblock Value-at-risk in portfolio optimization: properties and computational
  approach.
\newblock {\em The Journal of Risk}, 7(2):1.

\bibitem[Grechuk, 2015]{Grechuk2015}
Grechuk, B. (2015).
\newblock The center of a convex set and capital allocation.
\newblock {\em European Journal of Operational Research}, 243:628--636.

\bibitem[Grechuk et~al., 2013]{grechuk2013}
Grechuk, B., Molyboha, A., and Zabarankin, M. (2013).
\newblock Cooperative games with general deviation measures.
\newblock {\em Mathematical Finance}, 23(2):339--365.

\bibitem[Grechuk and Zabarankin, 2014]{GZ2014}
Grechuk, B. and Zabarankin, M. (2014).
\newblock Inverse portfolio problem with mean-deviation model.
\newblock {\em European Journal of Operational Research}, 234(2):481--490.

\bibitem[Grechuk and Zabarankin, 2016]{GZ2016}
Grechuk, B. and Zabarankin, M. (2016).
\newblock Inverse portfolio problem with coherent risk measures.
\newblock {\em European Journal of Operational Research}, 249(2):740--750.

\bibitem[Grechuk and Zabarankin, 2017]{grechuk2015synergy}
Grechuk, B. and Zabarankin, M. (2017).
\newblock Synergy effect of cooperative investment.
\newblock {\em Annals of Operations Research}, 249(1-2):409--431.

\bibitem[Grechuk and Zabarankin, 2018]{grechuk2017direct}
Grechuk, B. and Zabarankin, M. (2018).
\newblock Direct data-based decision making under uncertainty.
\newblock {\em European Journal of Operational Research}, 267(1):200--211.

\bibitem[Hunt et~al., 1992]{hunt1992prevalence}
Hunt, B.~R., Sauer, T., and Yorke, J.~A. (1992).
\newblock Prevalence: a translation-invariant “almost every” on
  infinite-dimensional spaces.
\newblock {\em Bulletin of the American mathematical society}, 27(2):217--238.

\bibitem[Kalkbrener, 2005]{kalkbrener2005}
Kalkbrener, M. (2005).
\newblock An axiomatic approach to capital allocation.
\newblock {\em Mathematical Finance}, 15(3):425--437.

\bibitem[Krokhmal et~al., 2002]{KPU2002}
Krokhmal, P., Palmquist, J., and Uryasev, S. (2002).
\newblock Portfolio optimization with conditional value-at-risk objective and
  constraints.
\newblock {\em Journal of Risk}, 4:43--68.

\bibitem[Lim et~al., 2011]{LIM2011163}
Lim, A.~E., Shanthikumar, J.~G., and Vahn, G.-Y. (2011).
\newblock Conditional value-at-risk in portfolio optimization: Coherent but
  fragile.
\newblock {\em Operations Research Letters}, 39(3):163 -- 171.

\bibitem[Lim et~al., 2010]{lim2010portfolio}
Lim, C., Sherali, H.~D., and Uryasev, S. (2010).
\newblock Portfolio optimization by minimizing conditional value-at-risk via
  nondifferentiable optimization.
\newblock {\em Computational Optimization and Applications}, 46(3):391--415.

\bibitem[Lim, 1981]{lim1981center}
Lim, T.~C. (1981).
\newblock The center of a convex set.
\newblock {\em Proceedings of the American Mathematical Society},
  81(2):345--346.

\bibitem[Litterman et~al., 2004]{litterman}
Litterman, R. et~al. (2004).
\newblock {\em Modern Investment Management: An Equilibrium Approach}.
\newblock John Wiley \& Sons.

\bibitem[Lwin et~al., 2017]{lwin2017mean}
Lwin, K.~T., Qu, R., and MacCarthy, B.~L. (2017).
\newblock Mean-var portfolio optimization: A nonparametric approach.
\newblock {\em European Journal of Operational Research}, 260(2):751--766.
\newblock The whole paper uses discrete distributions.

\bibitem[Meucci, 2005]{meucci2005}
Meucci, A. (2005).
\newblock {\em Risk and Asset Allocation}.
\newblock Springer, New York.

\bibitem[Palczewski, 2018]{Palczewski2018}
Palczewski, A. (2018).
\newblock {LP Algorithms for Portfolio Optimization: The PortfolioOptim
  Package}.
\newblock {\em {The R Journal}}, 10(1):308--327.

\bibitem[Palczewski and Palczewski, 2019]{palczewski2016}
Palczewski, A. and Palczewski, J. (2019).
\newblock {B}lack-{L}itterman model for continuous distributions.
\newblock {\em European Journal of Operational Research}, 273:708--720.

\bibitem[Phelps, 1974]{phelps1974dentability}
Phelps, R.~R. (1974).
\newblock Dentability and extreme points in {B}anach spaces.
\newblock {\em Journal of Functional Analysis}, 17(1):78--90.

\bibitem[Rockafellar, 1970]{rockafellar1970}
Rockafellar, R.~T. (1970).
\newblock {\em Convex Analysis}.
\newblock Princeton University Press.

\bibitem[Rockafellar et~al., 2006a]{rockafellar2006b}
Rockafellar, R.~T., Uryasev, S., and Zabarankin, M. (2006a).
\newblock Generalized deviations in risk analysis.
\newblock {\em Finance and Stochastics}, 10(1):51--74.

\bibitem[Rockafellar et~al., 2006b]{rockafellar2006c}
Rockafellar, R.~T., Uryasev, S., and Zabarankin, M. (2006b).
\newblock Optimality conditions in portfolio analysis with general deviation
  measures.
\newblock {\em Mathematical Programming, Ser. B}, 108:515--540.

\bibitem[Rockafellar et~al., 2007]{rockafellar2007}
Rockafellar, R.~T., Uryasev, S., and Zabarankin, M. (2007).
\newblock Equilibrium with investors using a diversity of deviation measures.
\newblock {\em Journal of Banking \& Finance}, 31(11):3251--3268.

\bibitem[Salinetti and Wets, 1979]{salinetti1979}
Salinetti, G. and Wets, R. J.-B. (1979).
\newblock On the convergence of sequences of convex sets in finite dimensions.
\newblock {\em SIAM Review}, 21(1):18--33.

\bibitem[Schneider, 1971]{schneider1971}
Schneider, R. (1971).
\newblock On {S}teiner points of convex bodies.
\newblock {\em Israel Journal of Mathematics}, 9(2):241--249.

\bibitem[Tasche, 2007]{tasche2007capital}
Tasche, D. (2007).
\newblock Capital allocation to business units and sub-portfolios: the euler
  principle.
\newblock {\em arXiv preprint arXiv:0708.2542}.

\bibitem[Vitale, 1985]{vitale1985steiner}
Vitale, R.~A. (1985).
\newblock The steiner point in infinite dimensions.
\newblock {\em Israel Journal of Mathematics}, 52(3):245--250.

\bibitem[Xia, 2004]{xia2004multi}
Xia, J. (2004).
\newblock Multi-agent investment in incomplete markets.
\newblock {\em Finance and Stochastics}, 8(2):241--259.

\end{thebibliography}

\appendices

\section{Law invariant selectors}\label{app:lawinv}

Section \ref{sec:robustsel} introduces a method for selecting a unique solution to the inverse portfolio optimization problem. The approach is based on the principle of robustness. Its advantage is that robust selector is always uniquely determined.

Here we discuss an alternative approach which is based on the principle of law-invariance. 
A law-invariant selector may not be unique for some deviation measures in which case the law-invariance fails to resolve the non-uniqueness of the inverse optimization problem.
However, it is financially and probabilistically natural and works in some important special cases.

\begin{definition}
A selector $f_\cD:{\L}^2(\Omega) \to {\L}^2(\Omega)$ is called \emph{law-invariant} if $E[Y_1f_\cD(X)]=E[Y_2f_\cD(X)]$ whenever pairs of r.v.s $(Y_1,X), (Y_2,X) \in \L^2(\Omega) \times \L^2(\Omega)$ have the same joint laws.
\end{definition}

A deviation measure $\cD$ is called \emph{law-invariant} if $\cD(X)=\cD(Y)$ whenever r.v.s $X$ and $Y$ have the same distribution. For example,  $\cvar_\alpha^\Delta$ (CVaR-deviation)  is law invariant for every $\alpha\in(0,1)$. Notice that not every deviation measure is law-invariant: a simple example of a non-law-invariant deviation measure can be constructed on $\Omega=\{\omega_1, \omega_2\}$, with $\prob[\omega_1]=\prob[\omega_2]=0.5$, and 
\begin{equation}\label{eqn:not_law_inv}
\cD(X) := \max\big\{X(\omega_1)-X(\omega_2), 2(X(\omega_2)-X(\omega_1))\big\}. 
\end{equation}

In the framework of uniform probability spaces, we prove below the existence, but not uniqueness, of a law-invariant selector. 

\begin{theorem}\label{lem:uniform}
If $\Omega$ is uniform, then there exists a law-invariant selector $f_\cD$ for every law-invariant deviation measure $\cD$.
\end{theorem}
\begin{proof}
It follows easily from Lemmas \ref{lem:lawinv_conv}, \ref{lem:selcond}, and \ref{lem:nonuniform} below.
\end{proof}

For non-uniform finite probability spaces, the notion of law-invariance as defined above is of little use for defining a unique selector, because, for example, on $\Omega=\{\omega_1, \omega_2\}$ with $\prob[\omega_1] \neq 0.5$, r.v.s $X$ and $Y$ have the same distribution if and only if $X=Y$, and, by definition, \emph{every} deviation measure, including \eqref{eqn:not_law_inv}, is law-invariant. For similar reasons, \emph{every} selector $f_\cD$ on such probability space is law-invariant. An appropriate extension of the notion of law-invariance to non-uniform probability spaces follows from results below.

An r.v. $X$ dominates r.v. $Y$ in second order stochastic dominance, denoted $X \succeq_2 Y$, if 
$$
\int\limits_{-\infty}^t F_X(x)dx \leq \int\limits_{-\infty}^t F_Y(x)dx, \quad \forall t\in{\mathbb R}.
$$
An r.v. $X$ dominates r.v. $Y$ in concave order, denoted $X \succeq_c Y$, if $E[X]=E[Y]$ and $X \succeq_2 Y$. A deviation measure $\cD$ is called consistent with concave order if $\cD(X) \leq \cD(Y)$ whenever $X \succeq_c Y$.
 
\begin{lemma}\label{lem:lawinv_conv}
If a deviation measure $\cD$ is consistent with the concave order, it is law-invariant. If $\Omega$ is uniform, the converse statement also holds.
\end{lemma}
\begin{proof}
The first statement is trivial, and the second one is well-known, but the proof is usually presented for atomless probability space, see \citet[Theorem 4.1]{dana2005representation}. For a discrete uniform $\Omega$, let r.v.s $X$ and $Y$ take values $x_1\leq\dots\leq x_N$ and $y_1\leq\dots\leq y_N$, respectively. Then $X \succeq_c Y$ is equivalent to
\begin{equation}\label{eqn:discreteconc}
\sum\limits_{i=1}^{k}x_i \geq \sum\limits_{i=1}^{k}y_i, \quad k=1, \dots, N,
\end{equation}
with equality for $k=N$. Let us prove that in this case $Y$ can be obtained from $X$ by a finite sequence of operations
\begin{multline}\label{eqn:oper}
(z_1, z_2, \dots, z_N) \to (z_1, \dots, z_{i-1}, z_i-d, z_{i+1}, \dots, z_{j-1}, z_j+d, z_{j+1}, \dots, z_N), \\ d>0, \,\, 1\leq i<j \leq N.
\end{multline}
The statement is trivial for $N=2$, and the case $N>2$ can be proved by induction.
If $\sum_{i=1}^{k}x_i = \sum_{i=1}^{k}y_i$ for some $k<N$, we can apply induction hypothesis to pair of r.v.s $X_1=(x_1, \dots, x_k)$ and $Y_1=(y_1, \dots, y_k)$, and separately to pair $X_2=(x_{k+1}, \dots, x_N)$ and $Y_2=(y_{k+1}, \dots, y_N)$, to conclude that there exists a sequence of operations (\ref{eqn:oper}) transforming $X_1$ to $Y_1$ and $X_2$ to $Y_2$, and hence $X$ to $Y$. Otherwise, apply operation (\ref{eqn:oper}) to $X$ with $i=1$, $j=N$, and 
$
d=\min_{k}\sum_{i=1}^{k}(x_i-y_i)>0,
$
to get $X = (x_1, x_2, \dots, x_N) \to (x_1-d, x_2, \dots, x_N+d) = (z_1, \dots, z_N) = Z.$
Then condition (\ref{eqn:discreteconc}) holds for $z_1, z_2, \dots, z_N$ in place of $x_1, x_2, \dots, x_N$, with equality for some $k<N$, hence $Z$ can be transformed to $Y$ by the argument above. 

Because operation (\ref{eqn:oper}) can only increase a law-invariant deviation measure $\cD$, $\cD(X) \leq \cD(Y)$ follows. 
\end{proof}

\begin{lemma}\label{lem:selcond}
If for any r.v.~$X \in \L^2(\Omega)$ the selector $f_\cD$ satisfies the condition
\begin{equation}\label{eqn:selcond}
Q(\omega_i) = Q(\omega_j) \quad \text{whenever} \quad X(\omega_i) = X(\omega_j), 
\end{equation}
where $Q=f_\cD(X)$, then it is law-invariant. If $\Omega$ is uniform, the converse statement also holds.
\end{lemma}
\begin{proof}
Condition \eqref{eqn:selcond} implies that $Q=g(X)$ for some function $g:{\mathbb R}\to{\mathbb R}$. Then $E[Y_1Q]=E[Y_1g(X)]=E[Y_2g(X)]=E[Y_2Q]$ whenever pairs of r.v.s $(Y_1,X)$ and $(Y_2,X)$ have the same joint law.

Conversely, let $\Omega$ be uniform and $X(\omega_i) = X(\omega_j)$. Then pairs of r.v.s $(I_i,X)$ and $(I_j,X)$ have the same joint law, where $I_i$ and $I_j$ are indicator functions for $\omega_i$ and $\omega_j$, respectively. If $f_\cD$ is law-invariant, this implies $Q(\omega_i)=N \cdot E[I_i Q] = N \cdot E[I_j Q] = Q(\omega_j)$, where $N = |\Omega|$, and \eqref{eqn:selcond} follows.
\end{proof}

Lemmas \ref{lem:lawinv_conv} and \ref{lem:selcond} imply that consistency with the concave ordering and \eqref{eqn:selcond} are appropriate extensions of the notion of law-invariance to non-uniform probability spaces for deviation measures and selectors, respectively. 

\begin{lemma}\label{lem:nonuniform}
For every deviation measure $\cD$, consistent with concave ordering, there exists a selector $f_\cD$ satisfying \eqref{eqn:selcond}.
\end{lemma}
\begin{proof}
Fix a r.v. $X$, select \emph{any} risk identifier $Q$ for $X$, and let $f_\cD(X):=E[Q|X]$. Then for all $Y\in{\L}^2(\Omega)$,
\[
E[(1-f_\cD(X)) Y] = E[(1-E[Q|X])Y] = E[(1-Q)(E[Y|X])]  \leq \cD(E[Y|X]) \leq \cD(Y),
\]
where the first inequality follows from $Q\in{\cal Q}$ and \eqref{eqn:devenvelopes}, while the second one follows from consistency of $\cD$ with concave ordering and the fact that $E[Y|X] \succeq_c Y$, see \citet[Corollary 2.61]{follmer2011stochastic}. Hence, $f_\cD(X)\in {\cal Q}$ by \eqref{eqn:devenvelopes2}. Because also $E[(1-f_\cD(X)) X]=E[(1-E[Q|X])X] = E[(1-Q)X]=\cD(X)$, $f_\cD(X)$ is in fact a risk identifier of $X$, and condition \eqref{eqn:selcond} trivially holds.   
\end{proof}

\begin{ex}\label{ex:cvarsel}
For CVaR-deviation $\cD={\cvar}_{\alpha}^\Delta$,  there exists a \emph{unique} selector satisfying \eqref{eqn:selcond}, and it is given by (see \cite{cherny2006})
\begin{equation}\label{eqn:cvarunique}
f_\cD(X)=Q_\alpha=
\begin{cases} 0, & X>-VaR_\alpha(X), \\ 
c_X, & X=-VaR_\alpha(X),\\ 
1/\alpha, & X<-VaR_\alpha(X),
\end{cases}
\end{equation}
where constant $c_X\in[0,1/\alpha]$ is such that $E[Q]=1$.
\end{ex}

\begin{ex}\label{ex:mixedcvarsel}
For mixed CVaR-deviation \eqref{eqn:mixedcvar}, there exists a \emph{unique} selector satisfying \eqref{eqn:selcond}, and it is of the form $f_\cD(X)=Q_\mu=\int_0^1 Q_\alpha\,\mu(d\alpha)$,
where $ Q_\alpha$ is given by \eqref{eqn:cvarunique}
(see \cite{cherny2006}).
\end{ex}

\begin{ex}\label{ex:madselmany}
For mean absolute deviation ${\rm MAD}(X)=\|X-EX\|_1$ (c.f. Example \ref{ex:2}), if $P(X=EX)>0$, there are infinitely many selectors satisfying \eqref{eqn:selcond}.
\end{ex}

Example \ref{ex:madselmany} demonstrates that imposing condition \eqref{eqn:selcond} may not be sufficient for specifying the unique solution, and, in this case, another method is required.

The following lemma demonstrates the consistency of concepts of robust and law-invariant selectors.

\begin{lemma}\label{lem:uni_rob}
Let $\Omega$ be uniform. Then, for every law-invariant deviation measure $\cD$, the corresponding robust selector $f_\cD$ is law-invariant.
\end{lemma}
\begin{proof}
Let $X\in\L^2(\Omega)$ and $1\leq i<j\leq N$ be such that $X(\omega_i)=X(\omega_j)$. Let $T:{\cal L}^2(\Omega)\to {\cal L}^2(\Omega)$ be a map interchanging indices $i$ and $j$, that is, for $Y=(y_1, \dots, y_N)$, 
$$
T(Y)=(y_1, \dots, y_{i-1}, y_j, y_{i+1}, \dots, y_{j-1}, y_i, y_{j+1}, \dots, y_N).
$$
Then $T(X)=X$. Let $A_\cD \subset {\mathbb R}^N$ be the set on which $\cD:{\mathbb R}^N \to {\mathbb R}$ is differentiable. This set has a full Lebesgue measure due to the convexity of $\cD$. By law-invariance, $\cD(T(Y))=\cD(Y), \, \forall\, Y$. Hence, for every $Y \in A_\cD$ such that $T(Y) \in A_\cD$, we have $T(f_\cD(Y))=f_\cD(T(Y))$.
Eqn. \eqref{eq:leblim} implies for any $Y \in \L^2 (\Omega) \equiv \mathbb{R}^N$,
\begin{align*}
T(f_\cD(Y)) &= T\big(\lim_{\epsilon\to 0} E[f_\cD(Y+e_\epsilon)]\big) = \lim_{\epsilon\to 0} E[T(f_\cD(Y+e_\epsilon))]\\
&= \lim_{\epsilon\to 0} E[(f_\cD(T(Y+e_\epsilon))] = f_\cD(T(Y)),
\end{align*}
where $e_{\epsilon}$ is uniformly distributed on the ball $B_\epsilon(0) \subset \mathbb{R}^N$ and the expectation operator integrates the randomness of $e_{\epsilon}$. Because $T(X)=X$, this implies $T(f_\cD(X))=f_\cD(X)$. Hence, \eqref{eqn:selcond} holds, and $f_\cD$ is law-invariant by Lemma \ref{lem:selcond}.
\end{proof}

In conclusion, this paper suggests two principles for choosing a unique selector, and hence a unique solution to the inverse optimization problem if the set $\cM$ has more than one point. One principle states that if $\cD$ is law-invariant, we should have $\mu_i=\mu_j$ in \eqref{eqn:mu_inv}, whenever pairs $(\hat r^{(i)}, \hat R^T x^M)$ and $(\hat r^{(j)}, \hat R^T x^M)$ have the same joint law. This principle is already sufficient to resolve the problem for CVaR-deviation, and, more generally, for mixed CVaR-deviation, but, in general, may not return a unique solution. Another principle postulates that selector should be ``robust'' as defined in Section \ref{sec:gradient}, and has an advantage that it \emph{always} returns a unique solution. However, its economic interpretation/justification is not as clear as for the law-invariance principle.

\end{document}